\documentclass[submission,copyright,creativecommons]{eptcs}
\providecommand{\event}{MFPS XXXVII} 
\usepackage{underscore}           
\usepackage{tikz}
\usepackage{tikz-cd}
\usepackage{stmaryrd}
\usepackage{amsmath,amssymb,amsthm}
\usepackage{enumerate}
\usepackage{hyperref}
\usepackage{multicol}
\usetikzlibrary{shapes}
\usetikzlibrary{positioning}
\usetikzlibrary{calc}
\usetikzlibrary{decorations.pathmorphing}

\tikzset{curve/.style={settings={#1},to path={(\tikztostart)
    .. controls ($(\tikztostart)!\pv{pos}!(\tikztotarget)!\pv{height}!270:(\tikztotarget)$)
    and ($(\tikztostart)!1-\pv{pos}!(\tikztotarget)!\pv{height}!270:(\tikztotarget)$)
    .. (\tikztotarget)\tikztonodes}},
    settings/.code={\tikzset{quiver/.cd,#1}
        \def\pv##1{\pgfkeysvalueof{/tikz/quiver/##1}}},
    quiver/.cd,pos/.initial=0.35,height/.initial=0}

\tikzset{tail reversed/.code={\pgfsetarrowsstart{tikzcd to}}}
\tikzset{2tail/.code={\pgfsetarrowsstart{Implies[reversed]}}}
\tikzset{2tail reversed/.code={\pgfsetarrowsstart{Implies}}}
\tikzset{no body/.style={/tikz/dash pattern=on 0 off 1mm}}

\usepackage{graphicx,float}
\usepackage{cite}
\newtheorem{theorem}{Theorem}[section]

\newtheorem{proposition}[theorem]{Proposition}
\newtheorem{lemma}[theorem]{Lemma}

\newtheorem{conjecture}[theorem]{Conjecture}

\theoremstyle{definition}
\newtheorem{definition}[theorem]{Definition}

\theoremstyle{remark}
\newtheorem{remark}[theorem]{Remark}

\newcommand{\EM}{\mathrm{EM}}
\newcommand{\EMs}{\mathrm{EM}_{\mathrm{s}}}

\newcommand{\inl}{\mathsf{inl}}
\newcommand{\inr}{\mathsf{inr}}
\renewcommand{\a}{\mathsf{a}}
\renewcommand{\b}{\mathsf{b}}
\newcommand{\point}{\bullet}
\newcommand{\mD}{\mathcal{D}}
\newcommand{\mP}{\mathcal{P}}
\newcommand{\id}{\mathrm{id}}
\newcommand{\Ms}{M^{\mathrm{s}}}
\newcommand{\etas}{\eta^{\mathrm{s}}}
\newcommand{\mus}{\mu^{\mathrm{s}}}
\newcommand{\musb}{\mu^{\mathrm{s},\mathsf{r}}}
\newcommand{\forget}{U^M}
\newcommand{\forgetMs}{U^{\Ms}}
\newcommand{\forgets}{U^M_{\mathrm{s}}}
\newcommand{\bigplus}{\mathop{\mbox{\Large$+$}}}
\newcommand{\xRightarrow}[1]{\stackrel{#1}{\Longrightarrow}}
\newcommand{\terms}[1]{T_{#1}}
\newcommand{\Alg}{\mathbf{A}}
\newcommand{\deltar}{\delta^{\mathsf{r}}}

\title{Semialgebras and Weak Distributive Laws} 
\author{Daniela Petrişan \institute{IRIF\\ Paris, France} \email{petrisan@irif.fr} \and
Ralph Sarkis \institute{ENS de Lyon\\ Lyon, France} \email{ralph.sarkis@ens-lyon.fr}}
\def\titlerunning{Semialgebras and Weak Distributive Laws}
\def\authorrunning{D. Petrişan \& R. Sarkis}

\begin{document}
\maketitle
\begin{abstract}
  Motivated by recent work on weak distributive laws and their
  applications to coalgebraic semantics, we investigate the
  algebraic nature of semialgebras for a monad. These are algebras
  for the underlying functor of the monad subject to the
  associativity axiom alone --- the unit axiom from the definition of
  an Eilenberg--Moore algebras is dropped. We prove that if the
  underlying category has coproducts, then semialgebras for a monad
  $M$ are in fact the Eilenberg--Moore algebras for a suitable monad
  structure on the functor $\id+M$, which we call the semifree monad
  $\Ms$. We also provide concrete algebraic presentations for
  semialgebras for the maybe monad, the semigroup monad and the
  finite distribution monad.
  
  A second contribution is characterizing the weak distributive laws
  of the form $MT\Rightarrow TM$ as strong distributive laws $\Ms
  T\Rightarrow T\Ms$ subject to an additional condition.
  \end{abstract}

\section{Introduction}\label{intro}
Distributive laws~\cite{Beck1969} are a standard approach to composing
monads. A distributive law of the monad $M$ over the monad $T$ is a
natural transformation $\lambda\colon MT \Rightarrow TM$ satisfying four axioms that stipulate its interactions with the units and
the multiplications of the two monads.
Unfortunately, such laws do not exist in several instances relevant
for the semantics of programming languages and for the composition of
computational effects. For example, there is no distributive law of
the finite distribution monad (used for modeling probabilistic choice)
over the powerset monad (used for modeling non-determinism),
see~\cite{Varacca06, Varacca-PhD}. Similarly, there is no distributive
law of the powerset monad over itself~\cite{Klin-Salamanca}. The
absence of such laws renders it difficult to reason in a compositional
way about systems combining non-determinism and probabilistic choice,
as emphasized in a series of works in domain theory~\cite{Keimel17,
  Mislove00} or in coalgebraic semantics~\cite{Bonchi19, Bonchi17b,
  MioSarkisVign21}.

Recently, weaker notions of distributive laws have been
considered~\cite{Street, Garner, Bohm}. In particular, Garner~\cite{Garner}
introduces a notion of weak distributive law $\lambda\colon MT \Rightarrow TM$
in which he drops the axiom related to the unit of the monad $M$,
while maintaining the three other axioms.  He proves that weak
distributive laws correspond under mild assumptions to a suitable
notion of weak liftings of the monad $T$ to the Eilenberg--Moore
algebras of $M$. The main example in \emph{loc. cit.} is a canonical
weak distributive law of the ultrafilter monad over the powerset
monad. The corresponding weak lifting of the powerset monad is the
Vietoris monad on the category of compact Hausdorff spaces.

In~\cite{GoyPetrisan20}, the same techniques were employed to find a
canonical weak distributive law of the finite distribution monad $\mD$
over the powerset monad $\mP$. The corresponding weak lifting of the
powerset is the convex powerset monad on the category of convex
algebras --- the Eilenberg--Moore algebras for $\mD$. This weak
distributive law enables a neat compositional approach to the
coalgebraic semantics of systems featuring both non-determinism and
probabilistic choice. This line of work was further extended to a
continuous setting in~\cite{Goy2021} where a weak distributive
law of the Vietoris monad over itself was
given. Furthermore,~\cite{Bonchi2021} provides a weak distributive law
between the powerset monad and the left-semimodule monad.

One of the technical ingredients in Garner's paper was to use the fact
that a weak distributive law $\lambda\colon MT\Rightarrow TM$ induces a
lifting of the monad $T$ to the category of semialgebras for the
monad $M$, where a semialgebra is a morphism $a\colon MX\to X$
satisfying only the associativity axiom from the standard definition
of an Eilenberg--Moore algebra. The weak lifting of the monad $T$ is
then obtained by splitting a certain idempotent in the category of
semialgebras for $M$.

For example, the weak law of~\cite{GoyPetrisan20} corresponds to a
lifting of $\mP$ to the category of semialgebras for the monad $\mD$,
which we refer to as convex semialgebras. A natural question arising
in these examples is what are concrete descriptions of
semialgebras. Even if it may seem surprising at first, dropping the
unit axiom alone does not considerably impair the structure given by
operations and equations.

Let us illustrate this with a simpler example, namely the semialgebras
for the maybe monad --- which maps $X$ to $X+\mathbf{1}$ where $\mathbf{1} = \{\star\}$. All the details will
be given in Section~\ref{exmps}. First, recall that algebras for the
maybe monad are pointed sets. Indeed, such an algebra is a map
$a\colon X+\mathbf{1}\to X$ satisfying the usual unit axiom
($a\circ\eta_X=\id_X$) and an associativity axiom. Already the unit
axiom entails that the first component of $a$ is the identity on $X$,
hence to give an algebra structure $a$ amounts to give one point
$\bullet\colon \mathbf{1}\to X$.

Let us consider now semialgebras $a\colon X+\mathbf{1}\to X$ for the maybe
monad. We no longer have that $a\circ\eta_X=\id_X$, but using the
associativity axiom, we can infer that $a\circ\eta_X$ is the first component of $a$, it is
idempotent, and furthermore, it preserves the point
$\bullet\colon \mathbf{1}\to X$. We can prove that the semialgebras for the
maybe monad are pointed sets equipped with an idempotent unary
operation that preserves the point. This algebraic theory corresponds
to a monad structure on the functor $X+X+\mathbf{1}$, that is, the coproduct of
the identity functor and  the underlying functor of the maybe monad.
Perhaps surprisingly, this result generalizes to arbitrary monads on
categories with coproducts.

\paragraph{Contributions.}
A first contribution of our paper is to unravel the algebraic nature
of semialgebras. Given a monad $M$ on a category $\mathbf{C}$ with
coproducts, we exhibit a monad structure $\Ms$ on the functor $\id+M$
(Theorem~\ref{thm:semifreemonad}) and we show that the category of
Eilenberg--Moore semialgebras for $M$ is isomorphic to the category of
Eilenberg--Moore algebras for $\Ms$ (Theorem~\ref{thm:isosemialgalg}).
In Section~\ref{exmps} we consider several examples: the maybe monad,
the semigroup monad and the finite distribution monad. In each case we
provide concrete algebraic presentations for the semialgebras.

A second contribution of our paper can be summed up in the slogan
``Weak distributive laws are strong''. Indeed, we can characterize the
liftings of the monad $T$ to the category of semialgebras for $M$
that correspond to weak distributive laws $MT\Rightarrow TM$, obtaining a
correspondence theorem akin to that of Beck,
Theorem~\ref{thm:wdlliftings}. Combining this result with the
isomorphism between semialgebras for $M$ and algebras for $\Ms$, we
prove a correspondence between weak distributive laws $MT\Rightarrow TM$ and
strong distributive laws $\Ms T\Rightarrow T\Ms$ that satisfy an additional
constraint (Theorem~\ref{thm:wdlanddl}).

\paragraph{Related work.} A similar question has been addressed
in~\cite{Hyland2021} in a 2-dimensional setting. In
\emph{loc. cit.} left semialgebras for a 2-monad were shown to be the
algebras for another 2-monad, which was obtained using a colax colimit
construction, see~\cite[Proposition~27]{Hyland2021}. 

\paragraph{Acknowledgements.} We thank Christine Tasson and Martin Hyland for discussion on these topics once upon a time in the wake of the pandemic when IRIF was still filled with people and thesis were defended on campus. This work was performed within the framework of the LABEX MILYON (ANR-10-LABX-0070) of Universit\'e de Lyon, within the program ``Investissements d'Avenir'' (ANR-11-IDEX-0007) operated by the French National Research Agency (ANR).

\section{Background}\label{background}
In this section, we present some definitions and results about monads, distributive laws, coproducts and universal algebra. We assume some familiarity with basic concepts in category theory --- \cite{Riehl,Awodey,MacLane71} are standard references.
\subsection{Monads and Distributive Laws}
\begin{definition}
    A \textbf{monad} on a category $\mathbf{C}$ is a triple comprised
    of an endofunctor $M: \mathbf{C} \rightarrow \mathbf{C}$ and two
    natural transformations $\eta: \id_{\mathbf{C}}\Rightarrow M$ and
    $\mu: M^2 \Rightarrow M$ called the \textbf{unit} and
    \textbf{multiplication} respectively that make
    \eqref{diag:unitmonad} and \eqref{diag:multmonad} commute. We
    refer to \eqref{diag:multmonad} as the associativity of $\mu$.\\
    \begin{minipage}{0.5\textwidth}
        \begin{equation}\label{diag:unitmonad}
            \begin{tikzcd}
                M & {M^2} & M \\
                & M
                \arrow["M\eta", from=1-1, to=1-2]
                \arrow["\mu", from=1-2, to=2-2]
                \arrow["{\id_{M}}"', from=1-1, to=2-2]
                \arrow["{\eta M}"', from=1-3, to=1-2]
                \arrow["{\id_{M}}", from=1-3, to=2-2]
            \end{tikzcd}
        \end{equation}
    \end{minipage}
    \begin{minipage}{0.5\textwidth}
        \begin{equation}\label{diag:multmonad}
            \begin{tikzcd}
                {M^3} & {M^2} \\
                {M^2} & M
                \arrow["\mu"', from=2-1, to=2-2]
                \arrow["M\mu"', from=1-1, to=2-1]
                \arrow["{\mu M}", from=1-1, to=1-2]
                \arrow["\mu", from=1-2, to=2-2]
            \end{tikzcd}
        \end{equation}
    \end{minipage}\\
\end{definition}
\begin{definition}
    Let $(M,\eta, \mu)$ be a monad on $\mathbf{C}$, an
    $M$\textbf{-algebra} is a pair $(X, x)$ consisting of an object
    $X$ and morphism $x : MX \rightarrow X$ in $\mathbf{C}$ such that
    \eqref{diag:algunit} and \eqref{diag:algmult} commute. We refer to
    \eqref{diag:algunit} as the unit axiom of $x$ and to
    \eqref{diag:algmult} as the associativity of $x$.\\
    \begin{minipage}{0.5\textwidth}
        \begin{equation}\label{diag:algunit}
            \begin{tikzcd}
                X \arrow[rd, "\id_X"'] \arrow[r, "\eta_X"] & MX \arrow[d, "x"] \\ & X
            \end{tikzcd}
        \end{equation}
    \end{minipage}
    \begin{minipage}{0.5\textwidth}
        \begin{equation}\label{diag:algmult}
            \begin{tikzcd}
                M^2X \arrow[d, "Mx"'] \arrow[r, "\mu_X"] & MX \arrow[d, "x"] \\
                MX \arrow[r, "x"']  & X
                \end{tikzcd}
        \end{equation}
    \end{minipage}
\end{definition}
\begin{definition}
    Given two $M$-algebras $(X, x)$ and $(Y, y)$, an $M$-algebra \textbf{homomorphism} $h: (X, x) \rightarrow (Y, y)$ is a morphism $h:X \rightarrow Y$ in $\mathbf{C}$ making \eqref{diag:alghom} commute.
    \begin{equation}\label{diag:alghom}
        \begin{tikzcd}
            MX \arrow[d, "x"'] \arrow[r, "Mh"] & MY
            \arrow[d, "y"] \\
            X \arrow[r, "h"'] & Y
        \end{tikzcd}
    \end{equation}
\end{definition}
For a monad $M$, the category of $M$-algebras and their homomorphisms is called the \textbf{Eilenberg--Moore} category of $M$ and denoted $\EM(M)$. We denote $\forget: \EM(M) \rightarrow \mathbf{C}$ the forgetful functor sending an $M$-algebra $(X,x)$ to $X$ and a homomorphism to its underlying morphism.
A morphism $x: MX \rightarrow X$ that satisfies associativity \eqref{diag:algmult} (but not necessarily the unit axiom \eqref{diag:algunit}) is called an $M$\textbf{-semialgebra}. We denote $\EMs(M)$ the category of $M$-semialgebras and their homomorphisms (defined as for $M$-algebras). We denote $\forgets: \EMs(M) \rightarrow \mathbf{C}$ the forgetful functor sending an $M$-semialgebra $(X,x)$ to $X$ and a homomorphism to its underlying morphism. 

Distributive laws between two monads, introduced in~\cite{Beck1969}, are
the category theoretic tool for composing monads and for representing
that the corresponding algebraic structures interact in a suitable way
via distributivity axioms.

\begin{definition}
    Let $(M, \eta^M, \mu^M)$ and $(T, \eta^T, \mu^T)$ be two monads, a natural transformation $\lambda: M T\Rightarrow TM$ is called a \textbf{(monad) distributive law of $M$ over $T$} if it makes the following diagrams commute.\\
    \begin{minipage}{0.5\textwidth}
        \begin{equation}\label{diag:dlunitM}
        \begin{tikzcd}
            & {T} \\
            {MT} && {TM}
            \arrow["{\eta^MT}"', from=1-2, to=2-1]
            \arrow["{T\eta^M}", from=1-2, to=2-3]
            \arrow["\lambda"', from=2-1, to=2-3]
        \end{tikzcd}
        \end{equation}
    \end{minipage}
    \begin{minipage}{0.5\textwidth}
        \begin{equation}\label{diag:dlunitMhat}
        \begin{tikzcd}
            & M \\
            {MT} && {TM}
            \arrow["{M\eta^T}"', from=1-2, to=2-1]
            \arrow["{\eta^TM}", from=1-2, to=2-3]
            \arrow["\lambda"', from=2-1, to=2-3]
        \end{tikzcd}
        \end{equation}
    \end{minipage}\\
    \begin{minipage}{0.5\textwidth}
        \begin{equation}\label{diag:dlmultM}
        \begin{tikzcd}
            {MMT} & {MTM} & {TMM} \\
            {MT} && {TM}
            \arrow["\lambda"', from=2-1, to=2-3]
            \arrow["{\mu^M T}"', from=1-1, to=2-1]
            \arrow["M\lambda", from=1-1, to=1-2]
            \arrow["{\lambda M}", from=1-2, to=1-3]
            \arrow["{T\mu^M}", from=1-3, to=2-3]
        \end{tikzcd}
        \end{equation}
    \end{minipage}
    \begin{minipage}{0.5\textwidth}
        \begin{equation}\label{diag:dlmultMhat}
        \begin{tikzcd}
            {MTT} & {TMT} & {TTM} \\
            {MT} && {TM}
            \arrow["\lambda"', from=2-1, to=2-3]
            \arrow["{M\mu^T}"', from=1-1, to=2-1]
            \arrow["{\lambda T}", from=1-1, to=1-2]
            \arrow["{T\lambda}", from=1-2, to=1-3]
            \arrow["{\mu^TM}", from=1-3, to=2-3]
        \end{tikzcd}
        \end{equation}
    \end{minipage}\\
\end{definition}
In~\cite{Street,Bohm}, the authors investigated weaker notions of monad distributive laws motivated by work on weak entwining operators. More recently, Garner~\cite{Garner} used one of them to exhibit the Vietoris monad on the category of compact
Hausdorff spaces as a weak lifting of the powerset monad. In the
sequel, we adopt the terminology and definitions of~\cite{Garner}.
\begin{definition}
    A \textbf{weak distributive law of $M$ over $T$} is a natural transformation $\lambda : MT \Rightarrow TM$ making \eqref{diag:dlunitMhat}, \eqref{diag:dlmultM} and \eqref{diag:dlmultMhat} commute (but not necessarily \eqref{diag:dlunitM}).
\end{definition}
\begin{definition}
    A \textbf{lifting} of $T$ to $\EM(M)$ is a monad $(\widetilde{T},\widetilde{\eta},\widetilde{\mu})$ on $\EM(M)$ such that the functor, unit and multiplication commute with the forgetful functor $\forget$, i.e., the following equations hold.
    \[\forget \circ \widetilde{T} = T \circ \forget \qquad \forget\widetilde{\eta} = \eta^T\forget \qquad \forget\widetilde{\mu} = \mu^T\forget\]
    A lifting of $T$ to $\EMs(M)$ is a monad $(\widetilde{T},\widetilde{\eta},\widetilde{\mu})$ satisfying $\forgets\widetilde{T} = T\forgets$, $\forgets\widetilde{\eta} = \eta^T\forgets$ and $\forgets\widetilde{\mu} = \mu^T\forgets$.
\end{definition}

A standard result that goes back to the work of Beck is the
correspondence between distributive laws and liftings to
Eilenberg--Moore categories. We recall below the transformation of a law into a lifting and its inverse as we will later use them in the weak setting.
\begin{proposition}\label{prop:isodllifts}
    Distributive laws $\lambda: MT \Rightarrow TM$ are in correspondence with liftings of $T$ to $\EM(M)$.
\end{proposition}
\begin{proof}
    Given a distributive law $\lambda: MT \Rightarrow TM$, we construct a lifting $\widetilde{T}$ sending $(X,x)$ to $(TX,Tx \circ \lambda_X)$. Its action on morphisms is determined by the equation $\forget \circ \widetilde{T} = T \circ \forget$, it must send $f$ to $Tf$. Given a lifting $(\widetilde{T},\widetilde{\eta},\widetilde{\mu})$, we construct a distributive law whose components are $\lambda_X= \widetilde{T}\mu^M_X \circ MT\eta^M_X$ where $\widetilde{T}\mu^M_X : MTMX \rightarrow TMX$ is the image of the free algebra $\mu^M_X:MMX \rightarrow MX$ under the functor $\widetilde{T}$.
\end{proof}
\begin{remark}
  There is a slight abuse of notation when writing $\widetilde{T}\mu^M_X$. To avoid any ambiguity, we will reserve writing $\widetilde{T}a$ only in
  the situations when $a$ is an algebra or semialgebra structure, but
  not when $a$ is a morphism of algebras.  Notice that if $f$ is a
  morphism of algebras, then the morphism $\widetilde{T}(f)$ is carried by
  $Tf$, since $\widetilde{T}$ is a lifting. In this situation we simply
  write $Tf$ instead of $\widetilde{T}(f)$.
\end{remark}
\subsection{Coproducts}
We recall here the definition of coproducts in order to present our notation and useful equations.
\begin{definition}
    Let $X$ and $Y$ be objects of $\mathbf{C}$, the coproduct of $X$ and $Y$ is an object $X+Y$ and morphisms $\inl^{X+Y}:X \rightarrow X+Y$ and $\inr^{X+Y}: Y \rightarrow X+Y$ such that for any pair of morphisms $k_X : X \rightarrow K$ and $k_Y: Y \rightarrow K$, there is a unique mediating morphism $!: X+Y \rightarrow K$ making \eqref{diag:upcoprod} commute.
    \begin{equation}\label{diag:upcoprod}
        \begin{tikzcd}
            X & {X+Y} & Y \\
            & K
            \arrow["{\inl^{X+Y}}", from=1-1, to=1-2]
            \arrow["{\inr^{X+Y}}"', from=1-3, to=1-2]
            \arrow["{!}", dashed, from=1-2, to=2-2]
            \arrow["{k_Y}"', from=1-1, to=2-2]
            \arrow["{k_Y}", from=1-3, to=2-2]
        \end{tikzcd}
    \end{equation}
    We may omit the superscripts on $\inl$ and $\inr$ when the codomain is clear from context. We denote $[k_X,k_Y]$ the unique morphism $X+Y \rightarrow K$ satisfying $[k_X,k_Y] \circ \inl = k_X$ and $[k_X,k_Y] \circ \inr = k_Y$. Given morphisms $f:X \rightarrow X'$ and $g: Y \rightarrow Y'$, we denote $f+g := [\inl^{X'+Y'} \circ f,\inr^{X'+Y'} \circ g]$. In the rest of the paper, we will often use the following easily derivable identities.\\
    \begin{minipage}{.49\linewidth}
      \begin{align}
        h \circ [k_X,k_Y] &= [h \circ k_X, h \circ k_Y]\label{eqn:coprod-1}\\
        \left[k_{X'},k_{Y'}\right] \circ (f+g) &= [k_{X'} \circ f,k_{Y'} \circ g] \label{eqn:coprod-2}
      \end{align}
\end{minipage}~~
\begin{minipage}{.49\linewidth}
  \begin{align}
    (f+g) \circ \inl^{X+Y} &= \inl^{X'+Y'} \circ f\label{eqn:coprod-3}\\
    (f+g) \circ \inr^{X+Y} &= \inr^{X'+Y'} \circ g \label{eqn:coprod-4}
  \end{align}
\end{minipage}
%
\end{definition}
When $\mathbf{C}$ has all coproducts, the category of endofunctors on $\mathbf{C}$ also does and coproducts are taken pointwise. Namely, given functors $F,G: \mathbf{C} \rightarrow \mathbf{C}$, $F+G$ sends an object $X$ to $FX+GX$ and a morphism $f$ to $Ff + Gf$. Moreover, we have
\[\inl^{F+G}_X = \inl^{FX+GX} \qquad \text{ and }\qquad \inr^{F+G}_X = \inr^{FX+GX}.\]

\subsection{Universal Algebra}
Here, we introduce just enough universal algebra to make use of the link between algebraic theories and monads in Section \ref{exmps} --- \cite{Bauer2019} is a longer gentle introduction to these notions.
\begin{definition}
    An \textbf{algebraic signature} is a set $\Sigma$ of operation symbols along with \textbf{arities} in $\mathbb{N}$, we denote $f: n \in \Sigma$ for an $n$-ary operation symbol $f$ in $\Sigma$. Given a set $X$, one constructs the set of $\Sigma$\textbf{-terms} with variables in $X$, denoted $\terms{\Sigma}(X)$ by iterating operations symbols:
    \begin{align*}
        \forall x \in X, &\ x \in \terms{\Sigma}(X)\\
        \forall t_1, \dots, t_n \in \terms{\Sigma}(X), f: n \in \Sigma, &\ f(t_1, \dots, t_n) \in \terms{\Sigma}(X).
    \end{align*}
    An \textbf{equation} over $\Sigma$ is a pair of $\Sigma$-terms over a set of indeterminate variables which we usually denote with an equality sign (e.g.: $s = t$ for $s,t \in \terms{\Sigma}(X)$ and $X$ is the set of variables). An \textbf{algebraic theory} is a tuple $(\Sigma, E)$ of a signature $\Sigma$ and a set $E$ of equations over $\Sigma$.

    Given an algebraic theory $(\Sigma,E)$, a $(\Sigma, E)$\textbf{-algebra} is a set $A$ along with operations $f^A: A^n \rightarrow A$ for all $f: n \in \Sigma$ (with the convention $A^0= \mathbf{1}$) such that the pairs of terms in $E$ are always equal when the operation symbols and variables are instantiated in $A$.\footnote{The operation symbol $f$ is always instantiated by $f^A$ and a variable can be instantiated by any element of $A$. For instance, suppose $(A,f^A, g^A)$ is a $(\Sigma,E)$-algebra and $f(x,g(y)) = g(y)$ is an equation in $E$, then for any $a, b\in A$, $f^A(a,g^A(b)) = g^A(b)$.}

    Given two $(\Sigma, E)$-algebras $A$ and $B$, a
    \textbf{homomorphism} between them is a map $h: A \rightarrow B$
    commuting with all operations in $\Sigma$, that is, $\forall f: n
    \in \Sigma, h\circ f^A = f^B \circ h^n$.\footnote{We write $h^n$
      for coordinatewise application of the map $h$ to vectors in
      $A^n$, i.e., $h^n(a_1,\dots,a_n) = (h(a_1),\dots,h(a_n))$.} The
    category of $(\Sigma,E)$-algebras and their homomorphisms is
    denoted $\Alg(\Sigma,E)$.
\end{definition}
We say that $(\Sigma,E)$ is an \textbf{algebraic presentation} for a monad $(M,\eta,\mu)$ if $\Alg(\Sigma,E) \cong \EM(M)$.

\section{Semifree Monad}\label{semifree}
In a 2-categorical setting, Hyland and Tasson gave an explicit
construction showing that the category of the so called left-semi
algebras of a 2-monad is monadic over the base category
\cite{Hyland2021}. This heavily relied on the 2-categorical
structure, particularly on the existence of some colax colimit. In
this section, we will show a similar result for semialgebras in a
1-categorical setting assuming only the existence of coproducts.

In the sequel, let $\mathbf{C}$ be a category with all coproducts and
$(M,\eta,\mu)$ be a monad on $\mathbf{C}$. We will define a monad
$\Ms$ and prove that $\EMs(M) \cong \EM(\Ms)$. We call $\Ms$ the
\textbf{semifree monad} on $M$ because of the similarity with the
definition of the algebraically free monad on a functor \cite{Kelly1980}. The functor
$\Ms$ is $\id_{\mathbf{C}}+M$, that is, the coproduct of
$\id_{\mathbf{C}}$ and $M$. The unit is $\etas := \inl^{\id+M}$ and
the multiplication is $\mus := [\id_{\id+M},\inr^{\id+M} \circ \mu
  \circ M[\eta,\id_M]]$, in a diagram:
\[\begin{tikzcd}
	&& \Ms\Ms \\
	{\id_{\mathbf{C}}+M} &&&& {M(\id_{\mathbf{C}}+M)} \\
	& {\id_{\mathbf{C}}+M} & M & MM
	\arrow["{\id}"', from=2-1, to=3-2]
	\arrow["\inr", from=3-3, to=3-2]
	\arrow["\mu", from=3-4, to=3-3]
	\arrow["{M[\eta,\id]}", from=2-5, to=3-4]
	\arrow["\inl", from=2-1, to=1-3]
	\arrow["\inr"', from=2-5, to=1-3]
	\arrow["\mus"{description}, dotted, from=1-3, to=3-2]
\end{tikzcd}\]
 Let us show that $(\Ms,\etas,\mus)$ is a monad.
\begin{lemma}\label{lem:unit-mus}
    The following diagram commutes.
    \begin{equation}\label{diag:unitdiagMs}
        \begin{tikzcd}
            \Ms & {\Ms\Ms} & \Ms \\
            & \Ms
            \arrow["\Ms\etas", from=1-1, to=1-2]
            \arrow["\mus", from=1-2, to=2-2]
            \arrow["{\id_{\Ms}}"', from=1-1, to=2-2]
            \arrow["{\etas \Ms}"', from=1-3, to=1-2]
            \arrow["{\id_{\Ms}}", from=1-3, to=2-2]
        \end{tikzcd}
    \end{equation}
\end{lemma}
\begin{proof}
    First, we show the left hand side commutes using the identities \eqref{eqn:coprod-1}-\eqref{eqn:coprod-4} and the left of \eqref{diag:unitmonad}.
    \begin{align*}
        \mus \circ \Ms\etas &= \mus \circ (\inl^{\id+M} +M\inl^{\id+M}) & \\
        &= [\id_{\id+M} \circ \inl^{\id+M}, \inr^{\id+M} \circ \mu \circ M[\eta,\id_M] \circ M\inl^{\id+M}] & \text{by  \eqref{eqn:coprod-2}}\\
        &= [\inl^{\id+M}, \inr^{\id+M} \circ \mu \circ M\eta] & \text{by } [\eta,\id_M]\circ\inl^{\id+M}=\eta \\
        &=  [\inl^{\id+M}, \inr^{\id+M}] & \text{by \eqref{diag:unitmonad}} \\ 
        &= \id_{\id+M} & 
    \end{align*}

    Next, we show the right hand side commutes.
    \begin{align*}
        \mus \circ \etas\Ms &= [\id_{\id+M},\inr^{\id+M} \circ \mu \circ M[\eta,\id_M]] \circ \inl^{\Ms + M\Ms}\\
        &= \id_{\id+M}
    \end{align*}
\end{proof}
\begin{lemma}\label{lem:assoc-mus}
    The multiplication $\mus$ is associative, i.e., $\mus \circ \mus\Ms = \mus \circ \Ms\mus$.
\end{lemma}
\begin{proof}
    Let us first expand both sides.
    \begin{align*}
        \text{L.H.S.} &= \mus \circ \mus\Ms & \\
        &= \mus \circ [\id_{\Ms+M\Ms}, \inr^{\Ms+M\Ms} \circ \mu\Ms \circ M[\eta\Ms, \id_{M\Ms}]] & \text{by def. }\mus\\
        &= [\mus, \mus \circ \inr^{\Ms+M\Ms} \circ \mu\Ms \circ M[\eta\Ms, \id_{M\Ms}]] & \text{by \eqref{eqn:coprod-1}} \\
        &= [\mus, \inr^{\id+M} \circ \mu \circ M[\eta,\id_M] \circ \mu\Ms \circ M[\eta\Ms, \id_{M\Ms}]] & \hspace{-6em}\text{by def. }\mus\text{ and } [k_X,k_Y] \circ \inr = k_Y\\    
        \text{R.H.S.} &= \mus \circ \Ms\mus &\\
        &= [\id_{\id+M}, \inr^{\id+M} \circ \mu \circ M[\eta,\id_M]] \circ (\mus + M\mus)&  \text{by def. }\mus\\
        &= [\mus,\inr^{\id+M} \circ \mu \circ M[\eta,\id_M] \circ M\mus]  & \text{by \eqref{eqn:coprod-2}}\\
        &= [\mus,\inr^{\id+M} \circ \mu \circ M([\eta,\id_M] \circ\mus)]  &\\
        &= [\mus,\inr^{\id+M} \circ \mu \circ M[[\eta,\id_M] \circ \id_{\id+M}, [\eta,\id_M] \circ \inr^{\id+M} \circ \mu \circ M[\eta,\id_M]] &  \text{by \eqref{eqn:coprod-1}} \\
        &= [\mus,\inr^{\id+M} \circ \mu \circ M[[\eta,\id_M], \mu \circ M[\eta,\id_M]] &\hspace{-6em} \text{by } [\eta,\id_M]\circ\inr^{\id+M}=\id_M
    \end{align*}
    Since the left component is the same, it is enough to prove the right component is also the same, i.e. that 
    \[\inr^{\id+M} \circ \mu \circ M[\eta,\id_M] \circ \mu\Ms \circ M[\eta\Ms, \id_{M\Ms}] = \inr^{\id+M} \circ \mu \circ M[[\eta,\id_M], \mu \circ M[\eta,\id_M]].\]
    We pave the following diagram.
    \begin{equation}\label{diag:pavingproofmonadMs}
\begin{tikzcd}
	{M(\id+M+M(\id+M))} &&&&& MM \\
	{MM(\id+M)} && MMM & MM && M \\
	{M(\id+M)} && MM & M && {\id+M}
	\arrow["{M[[\eta,\id_M],\mu \circ M[\eta,\id_M]]}", from=1-1, to=1-6]
	\arrow["\mu", from=1-6, to=2-6]
	\arrow["{\inr^{\id+M}}", from=2-6, to=3-6]
	\arrow[""{name=0, anchor=center, inner sep=0}, "{M[\eta\Ms,\id_{M\Ms}]}"', from=1-1, to=2-1]
	\arrow["\mu\Ms"', from=2-1, to=3-1]
	\arrow[""{name=1, anchor=center, inner sep=0}, "{M[\eta,\id_M]}"', from=3-1, to=3-3]
	\arrow[""{name=2, anchor=center, inner sep=0}, "\mu"', from=3-3, to=3-4]
	\arrow["{\inr^{\id+M}}"', from=3-4, to=3-6]
	\arrow[""{name=3, anchor=center, inner sep=0}, "{MM[\eta,\id_M]}", from=2-1, to=2-3]
	\arrow["{\mu M}"', from=2-3, to=3-3]
	\arrow[""{name=4, anchor=center, inner sep=0}, "M\mu"', from=2-3, to=2-4]
	\arrow["\mu"', from=2-4, to=3-4]
	\arrow[""{name=5, anchor=center, inner sep=0}, "M\mu"', from=2-3, to=1-6]
	\arrow["{\text{(a)}}"{description}, Rightarrow, draw=none, from=0, to=5]
	\arrow["{\text{(b)}}"{description}, Rightarrow, draw=none, from=3, to=1]
	\arrow["{\text{(c)}}"{description}, Rightarrow, draw=none, from=4, to=2]
\end{tikzcd}
    \end{equation}
    We show (a) below, (b) commutes by naturality of $\mu$ and (c) commutes by associativity of $\mu$. To prove (a) commutes, we can remove one application on $M$ on every morphism and starting with the bottom path, we have the following derivation.
    \begin{align*}
        \mu \circ M[\eta,\id_M] \circ [\eta\Ms,\id_{M\Ms}] &= [\mu \circ M[\eta,\id_M] \circ \eta\Ms, \mu \circ M[\eta,\id_M]] &  \text{by \eqref{eqn:coprod-1}}\\
        &= [\mu \circ \eta M \circ [\eta,\id_M], \mu \circ M[\eta,\id_M]] & \text{by nat. }\eta\\
        &= [[\eta,\id_M], \mu \circ M[\eta,\id_M]] & \text{by  \eqref{diag:unitmonad}} 
    \end{align*}
\end{proof}
Combining Lemmas~\ref{lem:unit-mus} and~\ref{lem:assoc-mus}, we obtain
the following result.
\begin{theorem}\label{thm:semifreemonad}
    The triple $(\Ms,\etas,\mus)$ is a monad.
\end{theorem}
Next, we show that this is the semifree monad on $M$.
\begin{theorem}\label{thm:isosemialgalg}
    There is an isomorphism $\EM(\Ms) \cong \EMs(M)$.
\end{theorem}
\begin{proof}
    First, note that any $\Ms$-algebra $\alpha : X+MX \rightarrow X$ must be of the form $\alpha = [\id_X,a]$ since $\id_X = \alpha \circ \etas_X = \alpha \circ \inl^{X+MX}$. Next, we claim the commutativity of the two following diagrams is equivalent, i.e., $[\id_X,a]$ is an $\Ms$-algebra if and only if $a$ is an $M$-semialgebra.\\
    \begin{minipage}{0.59\textwidth}
        \begin{equation}\label{diag:algMsiso}
            \begin{tikzcd}
            {X+MX+M(X+MX)} & {X+MX} \\
            {X+MX} & X
            \arrow["{[\id_X,a]+M[\id_X,a]}"', from=1-1, to=2-1]
            \arrow["{[\id_X,a]}"', from=2-1, to=2-2]
            \arrow["{\mus_X}", from=1-1, to=1-2]
            \arrow["{[\id_X,a]}", from=1-2, to=2-2]
            \end{tikzcd}
        \end{equation}
    \end{minipage}\begin{minipage}{0.40\textwidth}
        \begin{equation}\label{diag:semialgMiso}
            \begin{tikzcd}
                MMX & MX \\
                MX & X
                \arrow["Ma"', from=1-1, to=2-1]
                \arrow["a"', from=2-1, to=2-2]
                \arrow["{\mu_X}", from=1-1, to=1-2]
                \arrow["a", from=1-2, to=2-2]
            \end{tikzcd}
        \end{equation}
    \end{minipage}\\
    Since the bottom path of \eqref{diag:algMsiso} simplifies to $[[\id_X,a],a \circ M[\id_X,a]]$ and the top path simplifies to $[[\id_X,a],a \circ \mu_X \circ M[\eta_X,\id_{MX}]]$, we infer that the left square commutes if and only if 
    \begin{equation}\label{eqn:equivMsalg}
        a \circ M[\id_X, a] = a \circ \mu_X \circ M[\eta_X,\id_{MX}].
    \end{equation}
    Next, if \eqref{eqn:equivMsalg} holds, we pre-compose by $M\inr^{X+MX}$ and find that 
    \[a \circ Ma =a \circ M[\id_X,a] \circ M\inr^{X+MX} = a \circ \mu_X \circ M[\eta_X,\id_{MX}] \circ M\inr^{X+MX}= a \circ \mu_X.\]
    Conversely, if \eqref{diag:semialgMiso} commutes ($a \circ Ma = a \circ \mu_X$), we can derive the following two equalities.
    \begin{gather*}
        a \circ Ma \circ M\eta_X = a \circ \mu_X \circ M\eta_X = a\\
        a \circ \eta_X \circ a = a \circ Ma \circ \eta_{MX} = a \circ \mu_X \circ \eta_{MX} = a,
    \end{gather*}
    which lead to the following derivation showing \eqref{eqn:equivMsalg} holds.
    \begin{align*}
        a \circ M[\id_X,a] &= a \circ Ma \circ M\eta_X \circ M[\id_X,a] & \text{by 1st eqn. above} \\
        &= a \circ M[ a \circ \eta_X, a \circ \eta_X \circ a] &\text{by \eqref{eqn:coprod-1}}\\
        &= a \circ M[a \circ \eta_X, a] & \text{by 2nd eqn. above}\\
        &= a \circ Ma \circ  M[\eta_X,\id_{MX}]&\text{by \eqref{eqn:coprod-1}}\\
        &= a \circ \mu_X \circ M[\eta_X,\id_{MX}] &\text{by \eqref{diag:algmult}}
    \end{align*}

    We have shown that the assignments $[\id_X,a] \mapsto a$ and $a \mapsto [\id_X,a]$ are well-typed, and they are clearly inverses. It is left to show they are functorial. It is enough to show that a homomorphism between $[\id_X,a]$ and $[\id_X, b]$ is a homomorphism between $a$ and $b$ and vice versa.

    Suppose, $f \circ [\id_X,a] = [\id_X, b] \circ (f+Mf)$, then pre-composing with $\inr^{X+MX}$ yields $f \circ a = b \circ Mf$. Conversely, if $f \circ a = b \circ Mf$, we have 
    \[f \circ [\id_X,a] = [f,f \circ a] = [f, b \circ Mf]= [\id_X, b] \circ (f+Mf).\]
    We conclude the desired isomorphism.
\end{proof}

\begin{remark}\label{rem:idempot}
  From the proof of the above theorem, we can also infer the following
  fact, which plays an important role in~\cite{Garner} and in the
  concrete presentations by operations and equations provided in
  Section~\ref{exmps}. Given a semialgebra $a \colon MX\to X$ in
  $\EMs(M)$, we have that $a\circ\eta_X\colon X\to X$ is
  idempotent, i.e., $(a \circ \eta_X) \circ (a \circ \eta_X) = a \circ \eta_X$.
\end{remark}
\section{Weak Distributive Laws are Strong}\label{dllifts}
In this section, we will give an analogue to Proposition \ref{prop:isodllifts} in the case of weak distributive laws. Then, we will use the isomorphism $\EMs(M) \cong \EM(\Ms)$ to obtain a correspondence between weak distributive laws $MT \Rightarrow TM$ and (strong) distributive laws $\Ms T \Rightarrow T\Ms$ satisfying an additional constraint.

Let $(M,\eta^M,\mu^M)$ and $(T,\eta^T,\mu^T)$ be monads on a category $\mathbf{C}$ with all coproducts.

If we try to apply the same construction from Proposition
\ref{prop:isodllifts} to a \textit{weak} distributive law $\lambda: MT
\Rightarrow TM$, we quickly encounter a problem when proving that for
any $M$-algebra $x:MX \rightarrow X$, $\widetilde{T}x := Tx \circ
\lambda_X$ is an $M$-algebra. Indeed, showing that
\eqref{diag:algunit} commutes relies on the derivation \[Tx \circ
\lambda_X \circ \eta^M_{TX} = Tx \circ T\eta^M_X = T\id,\] which needs
\eqref{diag:dlunitM} to commute, hence $\lambda$ to be strong. We will
see in Theorem~\ref{thm:wdlliftings} that this is the only obstacle to
construct a lifting, namely that the lifting $\widetilde{T}$ is now on
$\EMs(M)$. Notice that the construction of the lifting $\widetilde{T}$
to semialgebras obtained from a weak distributive law also appears in
the proof of~\cite[Proposition~13]{Garner}. However,
Theorem~\ref{thm:wdlliftings} takes this further and characterizes the
liftings $\widetilde{T}$ on $\EMs(M)$ that correspond to weak
distributive laws.

Indeed, in the other direction, there is another issue when showing
that the transformation obtained from a lifting on $\EMs(M)$ makes
\eqref{diag:dlmultMhat} commute. In the setting of Proposition
\ref{prop:isodllifts}, we can use the fact that $\widetilde{T}\mu^M_X$
satisfies the unit axiom of an $M$-algebra, but this is not
necessarily the case here. Therefore, we must add the restriction
\eqref{eqn:hypothesislifting} to the liftings on $\EMs(M)$ to obtain
the correspondence in Theorem~\ref{thm:wdlliftings}.

Let us try to provide some intuition behind the
condition~\eqref{eqn:hypothesislifting} featured in this
theorem. Recall from Remark~\ref{rem:idempot} that, whenever
$\alpha\colon MA\to A$ is a semialgebra for $M$, then
$\alpha\circ\eta_A$ is an idempotent, that we will denote here by
$\mathsf{a}\colon A\to A$. Also recall that $\widetilde{T}\alpha\colon
MTA\to TA$ is the semialgebra obtained by applying $\widetilde{T}$ to
$\alpha$. Then condition~\eqref{eqn:hypothesislifting} roughly means that,
given any term $t(x_1,\ldots, x_n)$ in $MTA$, we have that
$\widetilde{T}\alpha(t(x_1,\ldots,
x_n))=\widetilde{T}\alpha(t(\mathsf{a}x_1,\ldots, \mathsf{a}x_n))$,
that is, applying the idempotent $\mathsf{a}$ to the leaves of any term in $MTA$
does not change the evaluation of that term under the semialgebra
$\widetilde{T}\alpha$.

\begin{theorem}\label{thm:wdlliftings}
    Weak distributive laws $\lambda: MT \Rightarrow TM$ are in correspondence with liftings $(\widetilde{T},\widetilde{\eta},\widetilde{\mu})$ of $T$ to $\EMs(M)$ such that for any $M$-semialgebra $\alpha:MA \rightarrow A$, 
    \begin{equation}\label{eqn:hypothesislifting}
        \widetilde{T}\alpha = \widetilde{T}\alpha \circ MT\alpha \circ MT\eta^M_A.
    \end{equation}
\end{theorem}
\begin{proof}
    A full proof is given in the appendix.  The correspondence is
    given by the same assignments as in Proposition
    \ref{prop:isodllifts} and in fact, the proof will closely follow
    the one in  \cite[Chapter~3]{Tanakathesis} except for some minor
    steps mentioned above which rely on the distributive law being
    strong and the objects of $\EM(M)$ satisfying the unit axiom.
\end{proof}

Next, we combine the characterization of liftings to semialgebras
coming from weak distributive laws and the characterization of
semialgebras as algebras for the semifree monad to obtain Theorem
\ref{thm:wdlanddl}. For that, we need to describe how the isomorphism
$\EMs(M) \cong \EM(\Ms)$ leads to a correspondence between liftings of
$T$ to these categories.
\begin{lemma}\label{lem:liftingsemems}
    Liftings of $T$ to $\EMs(M)$ are in correspondence with liftings
    of $T$ to $\EM(\Ms)$.
\end{lemma}
\begin{proof}
    Note that the isomorphism $S:\EMs(M) \cong \EM(\Ms):S^{-1}$
    described in Theorem~\ref{thm:isosemialgalg} commutes with the
    forgetful functors $\forgets$ and $\forgetMs$, namely, the
    following diagram commutes.
    \begin{equation}\label{diag:isocommuteforget}
        \begin{tikzcd}
            {\EMs(M)} && {\EM(\Ms)} \\
            & {\mathbf{C}}
            \arrow["\forgets"', from=1-1, to=2-2]
            \arrow["\forgetMs", from=1-3, to=2-2]
            \arrow["S", shift left=1, from=1-1, to=1-3]
            \arrow["{S^{-1}}", shift left=1, from=1-3, to=1-1]
        \end{tikzcd}
    \end{equation}
    Therefore, it is straightforward to check that if $\widetilde{T}$ is a lifting of $T$ on $\EMs(M)$, then $\widetilde{T}^{\mathrm{s}}:= S\widetilde{T}S^{-1}$ is a lifting of $T$ on $\EM(\Ms)$ and conversely if $\widetilde{T}^{\mathrm{s}}$ is a lifting of $T$ on $\EM(\Ms)$, then $S^{-1}\widetilde{T}^{\mathrm{s}}S$ is a lifting of $T$ on $\EMs(M)$. 
\end{proof}
\begin{theorem}\label{thm:wdlanddl}
    Weak distributive laws $\lambda: MT \Rightarrow TM$ are in correspondence with distributive laws $\delta: \Ms T \Rightarrow  T\Ms$ satisfying
    \begin{equation}\label{eqn:restrictdlMs}
        \delta \circ \inr^{T+MT} = T[\inr^{\id+M} \circ \eta^M, \inr^{\id+M}] \circ \delta \circ \inr^{T+MT}.
    \end{equation}
\end{theorem}
\begin{proof}
    First, note that a distributive law $\delta:\Ms T \Rightarrow T\Ms$ satisfies an instance of \eqref{diag:dlunitM} saying $\delta \circ \etas T = T\etas$. Thus, we obtain that $\delta = [T\inl^{\id+M}, \deltar]$ for some natural transformation $\deltar: MT \Rightarrow T\Ms$. Then, \eqref{eqn:restrictdlMs} can be simplified to 
    \begin{equation}\label{eqn:restrictdlMssimple}
        \deltar =T[\inr^{\id+M} \circ \eta^M, \inr^{\id+M}] \circ \deltar.
    \end{equation}
    To further lighten notation, we let $\musb := \inr^{\id+M} \circ \mu^M \circ M[\eta^M,\id_M]$ so that $\mus = [\id_{\id+M},\musb]$.

    In the forward direction, we start with a weak distributive law $\lambda:MT \Rightarrow TM$ and let $\widetilde{T}$ be the lifting obtain from Theorem \ref{thm:wdlliftings}. Through the isomorphism $\EMs(M) \cong \EM(\Ms)$, we obtain a lifting $\widetilde{T}^{\mathrm{s}}$ of $T$ to $\EM(\Ms)$, it sends an $\Ms$-algebra $[\id_X,x]$ to $[\id_{TX},\widetilde{T}x]$. Next, by the correspondence in Proposition \ref{prop:isodllifts}, we get a distributive law $\delta: \Ms T \Rightarrow T\Ms$ whose components are 
    \[\delta_X = \Ms TX \xrightarrow{\Ms T\inl^{X+MX}} \Ms T\Ms X \xrightarrow{\widetilde{T}^{\mathrm{s}}\mus_X} T\Ms X.\]
    Since we know the first component is $T\inl^{\id+M}$, we are more interested in the second component which we find to be \begin{align*}
        \deltar_X &= \widetilde{T}^{\mathrm{s}}\mus_X \circ \Ms T\inl^{X+MX} \circ \inr^{TX+MTX}\\
        &= [\id_{T(X+MX)}, \widetilde{T}\musb_X] \circ (T\inl^{X+MX}+MT\inl^{X+MX}) \circ \inr^{TX+MTX}\\
        &= [\id_{T(X+MX)}, \widetilde{T}\musb_X] \circ \inr^{T\Ms X+MT\Ms X} \circ MT\inl^{X+MX}\\
        &= \widetilde{T}\musb_X \circ MT\inl^{X+MX}.
    \end{align*}
    Now, using the fact that $\widetilde{T}\musb_X = T\musb_X \circ \lambda_X$ and the definition of $\musb$, we find that 
    \begin{align*}
        &T[\inr^{X+MX} \circ \eta^M_X, \inr^{X+MX}] \circ \deltar_X\\
        &\hspace{6em}= T[\inr^{X+MX} \circ \eta_X^M, \inr^{X+MX}] \circ \widetilde{T}\musb_X \circ MT\inl^{X+MX}\\
        &\hspace{6em}= T[\inr^{X+MX} \circ \eta_X^M, \inr^{X+MX}] \circ T\musb_X \circ \lambda_X \circ MT\inl^{X+MX}\\
        &\hspace{6em}= T[\inr^{X+MX} \circ \eta_X^M, \inr^{X+MX}] \circ T\inr^{X+MX} \circ T\mu^M_X \circ TM[\eta^M_X,\id_{MX}] \circ \lambda_X \circ MT\inl^{X+MX}\\
        &\hspace{6em}= T\inr^{X+MX} \circ T\mu^M_X \circ TM[\eta^M_X,\id_{MX}] \circ \lambda_X \circ MT\inl^{X+MX}\\
        &\hspace{6em}= \widetilde{T}\musb_X \circ MT\inl^{X+MX}\\
        &\hspace{6em}= \deltar_X.
    \end{align*}

    In the opposite direction, we start with a distributive law $\delta = [T\inl^{\id+M},\deltar]:\Ms T \Rightarrow T\Ms$ which is sent (using Proposition \ref{prop:isodllifts}) to a lifting $\widetilde{T}^{\mathrm{s}}$ of $T$ on $\EM(\Ms)$ that sends an $\Ms$-algebra $[\id_X,x]:\Ms X \rightarrow X$ to \[T[\id_X,x] \circ [T\inl^{X+MX},\deltar_X] = [\id_{TX},T[\id_X,x] \circ \deltar_X].\]
    Using Lemma \ref{lem:liftingsemems}, we obtain a lifting $\widetilde{T}$ on $\EMs(M)$ sending an $M$-semialgebra $x:MX \rightarrow X$ to $T[\id_X,x] \circ \deltar_X$. After showing this lifting satisfies \eqref{eqn:hypothesislifting} below, we can use Theorem \ref{thm:wdlliftings} to obtain the weak distributive law $\lambda$.
    \begin{align*}
        \widetilde{T}x \circ MTx \circ MT\eta^M_X &= T[\id_X,x] \circ \deltar_X \circ MTx \circ MT\eta^M_X\\
        &= T[\id_X,x] \circ T(x+Mx) \circ T(\eta^M_X+M\eta^M_X)\circ \deltar_X &\text{by nat. of $\deltar$}\\
        &= T[x \circ \eta^M_X, x \circ Mx \circ M\eta^M_X]\circ \deltar_X\\
        &= T[x \circ \eta^M_X, x \circ \mu^M_X \circ M\eta^M_X]\circ \deltar_X &\text{by assoc. of $x$}\\
        &= T[x \circ \eta^M_X, x]\circ \deltar_X &\text{by \eqref{diag:multmonad}}\\
        &= T[\id_X,x] \circ T[\inr^{X+MX} \circ \eta^M_X, \inr^{X+MX}]\circ \deltar_X \\
        &= T[\id_X,x] \circ \deltar_X &\text{by \eqref{eqn:restrictdlMssimple}}\\
        &= \widetilde{T}x.
    \end{align*}
\end{proof}
The results of Section \ref{dllifts} can be summarized as follows.
\begin{align*}
    \{\lambda:MT \xRightarrow{\text{w.d.l.}} TM\}
    &\stackrel{\text{Thm \ref{thm:wdlliftings}}}{\longleftrightarrow}
    \{\widetilde{T}:\EMs(M) \xrightarrow{\text{lifts }T}\EMs(M) \text{ satisfying \eqref{eqn:hypothesislifting}}\}\\
    &\stackrel{\text{Lem \ref{lem:liftingsemems}}}{\longleftrightarrow}
    \{\widetilde{T}^{\mathrm{s}}:\EM(\Ms) \xrightarrow{\text{lifts }T}\EM(\Ms)  \text{ such that $S^{-1}\widetilde{T}^{\mathrm{s}}S$ satisfies \eqref{eqn:hypothesislifting}}\}\\
    &\stackrel{\text{Thm \ref{thm:wdlanddl}}}{\longleftrightarrow}
    \{\delta:\Ms T \xRightarrow{\text{d.l.}}T\Ms \text{ satisfying \eqref{eqn:restrictdlMssimple}}\}
\end{align*}

\section{Examples}\label{exmps}
In this section, we give algebraic presentations of three semifree monads.

\subsection{Maybe Monad}
The maybe monad is defined on the functor $-+\mathbf{1}: \mathbf{Set} \rightarrow \mathbf{Set}$ where $\mathbf{1} = \{\star\}$. The unit and multiplication have components given by \[\eta_X = \inl^{X+\mathbf{1}} :X \rightarrow X+\mathbf{1} \qquad \text{ and } \qquad \mu_X = [\id_{X+\mathbf{1}}, \inr^{X+\mathbf{1}}] : X+\mathbf{1}+\mathbf{1} \rightarrow X+\mathbf{1}.\]

By Theorems \ref{thm:semifreemonad} and \ref{thm:isosemialgalg}, we know that the semifree monad for $-+\mathbf{1}$ is a monad on the functor $X \mapsto X+X+\mathbf{1}$ with unit $\etas_X = \inl^{X+(X+\mathbf{1})}$ and multiplication $\mus_X = [\id_X,\inr^{X+(X+\mathbf{1})} \circ [\id_{X+\mathbf{1}},\inr^{X+\mathbf{1}}] \circ ([\inl^{X+\mathbf{1}}, \id_{X+\mathbf{1}}]+\id_{\mathbf{1}})]$. This is very opaque and it does not help us understand the semialgebras for $-+\mathbf{1}$.

An alternative way to see these semialgebras is through the point of view of universal algebra. The theory of pointed sets containing a single constant $\point: 0$ with no equations is an algebraic presentation of the maybe monad. Briefly, this is because a $(-+\mathbf{1})$-algebra $a: X+\mathbf{1} \rightarrow X$ satisfies $\id_X = a \circ \eta_X = a \circ \inl$, thus $a = [\id_X,\point]$ where $\point: \mathbf{1} \rightarrow X$ is the constant. However, in a semialgebra, $a \circ \eta_X$ is only required to satisfy $a \circ \eta_X \circ a = a$ (see proof of Theorem \ref{thm:isosemialgalg}). Therefore, denoting $a = [\a,\point]$, we find (by pre-composing $a \circ \eta_X \circ a = a$ with $\inl$ and $\inr$) that $\a$ is idempotent and $\a(\point) = \point$. We infer that a $(-+\mathbf{1})$-semialgebra can be presented with an idempotent unary operation and a constant preserved by the idempotent.
\begin{theorem}\label{thm:pressemimaybe}
    Let $\Sigma^{\mathrm{s}}_{+\mathbf{1}} = \{\a:1,\point:0\}$ and $E^{\mathrm{s}}_{+\mathbf{1}} = \{\a\a x = \a x,\a\point = \point\}$, then $\Alg(\Sigma^{\mathrm{s}}_{+\mathbf{1}},E^{\mathrm{s}}_{+\mathbf{1}}) \cong \EMs(M)$.
\end{theorem}
\begin{proof}
    A semialgebra for the maybe monad is a function $a:X+\mathbf{1} \rightarrow X$ satisfying $a \circ (a+\id_{\mathbf{1}}) = a \circ [\id_{X+\mathbf{1}},\inr^{X+\mathbf{1}}]$. Pre-composing by $\inl^{X+\mathbf{1}}$ and $\inr^{X+\mathbf{1}}$, we find that, equivalently, $a$ satisfies $a \circ \inl \circ a = a$.

    Now, given a $(-+\mathbf{1})$-semialgebra $a:X+\mathbf{1} \rightarrow X$, we define $\a := a \circ \inl: X \rightarrow X$ and $\point := a \circ \inr: \mathbf{1} \rightarrow X$. These operations satisfy the equations in $E^{\mathrm{s}}_{+\mathbf{1}}$ by the following derivations.
    \begin{gather*}
        \a \circ \a = a \circ \inl \circ a \circ \inl = a \circ \inl = \a \\
        \a \circ \point = a \circ \inl \circ a \circ \inr = a \circ \inr = \point
    \end{gather*}
    Conversely, given $\a:X \rightarrow X$ and $\point: \mathbf{1} \rightarrow X$ satisfying $\a \circ \a = \a$ and $\a \circ \point = \point$, we define $a:= [\a,\point]: X+\mathbf{1} \rightarrow X$. To verify $[\a,\point]$ is a $(-+\mathbf{1})$-semialgebra, it is enough to check that $[\a,\point] \circ \inl \circ [\a,\point] = [\a,\point]$. This follows like so:
    \[[\a,\point] \circ \inl \circ [\a,\point] = \a \circ [\a, \point] = [\a \circ \a, \a \circ \point] = [\a, \point].\]

    These operations are clearly inverses, and we are left to show that they are functorial. Suppose $f:X \rightarrow Y$ is a homomorphism from $a$ to $b$ (i.e. $f \circ a = b \circ (f+\id_{\mathbf{1}})$), then 
    \begin{gather*}
        f \circ \a = f \circ a \circ \inl = b \circ (f+\id_{\mathbf{1}}) \circ \inl = b \circ \inl \circ f = \b \circ f\\
        f \circ \point^a = f \circ a \circ \inr = b \circ (f+\id_{\mathbf{1}}) \circ \inr = b \circ \inr = \point^b.
    \end{gather*}
    Conversely, suppose $f \circ \a = \b \circ f$ and $f \circ \point^a = \point^b$, then
    \[f \circ [\a,\point^a] = [f \circ \a, f \circ \point^a] = [\b \circ f, \point^b] = [\b,\point^b] \circ (f+\id_{\mathbf{1}}).\] 
\end{proof}

\subsection{Semigroup Monad}
The semigroup (or non-empty lists) monad $(-)^+ : \mathbf{Set} \rightarrow \mathbf{Set}$ sends $X$ to $X^+$ the set of non-empty finite words over $X$ (we denote them with lists e.g.: $[x_1,x_2,x_3]$). The unit and multiplication are given by $\eta_X: X \rightarrow X^+ = x \mapsto [x]$ and \[\mu_X: (X^+)^+ \rightarrow X^+ = [[x_{1,1},\dots, x_{n_1,1}], \dots, [x_{k,1},\dots, x_{n_k,k}]] \mapsto [x_{1,1},\dots, x_{n_1,1}, x_{2,1} \dots, x_{n_{k-1},k-1},x_{k,1},\dots, x_{n_k,k}].\]
This monad is presented by the theory of semigroups which contains a binary operation with an associativity equation. We will not bother working out what the semifree monad for $(-)^+$ is, and we give its algebraic presentation at once.
\begin{theorem}\label{thm:pressemisemigroup}
    Let $\Sigma^{\mathrm{s}}_+ = \{\a : 1, \cdot : 2\}$ and $E^{\mathrm{s}}_+$ contain 
    \begin{align*}
        \a\a x &= \a x\\
        \a(x\cdot y) &= x \cdot y\\
        \a x \cdot \a y &= x\cdot y\\
        (x\cdot y) \cdot z &= x \cdot (y\cdot z),
    \end{align*}
    then $\Alg(\Sigma^{\mathrm{s}}_+, E^{\mathrm{s}}_+) \cong \EMs((-)^+)$.
\end{theorem}
\begin{proof}
    The structure of the proof is exactly the same as for Theorem \ref{thm:pressemimaybe} (all details are in the appendix). A $(-)^+$-semialgebra $a:X^+ \rightarrow X$ is sent to $(X,\a:0,\cdot:2)$ where  $\a:= a[-]$ and $\cdot := a[-,-]$. A $((\Sigma^{\mathrm{s}}_+, E^{\mathrm{s}}_+))$-algebra $(X,\a,\cdot)$ is sent to $a: X^+ \rightarrow X = [x_1,\dots, x_n] \mapsto \a x_1 \cdots \a x_n$ which is well-defined by associativity of $\cdot$. There is a technical difficulty which requires us to prove (by induction) that the equations in $E^{\mathrm{s}}_+$ imply
    \[\forall n\geq 2, \qquad \a x_1 \cdots \a x_n = x_1 \cdots x_n = \a(x_1\cdots x_n).\]
\end{proof}
\subsection{Distribution Monad}
The distribution monad $\mD: \mathbf{Set} \rightarrow \mathbf{Set}$ sends $X$ to $\mD X$ the set of finitely supported distributions on $X$, i.e., 
\[\mD(X) := \{\varphi \in [0,1]^X \mid \sum_{x \in X} \varphi(x) = 1 \text{ and } \varphi(x) \neq 0 \text{ for finitely many $x$s}\}.\]
Its unit and multiplication are given by $\eta_X = x \mapsto 1x$, where $1x$ is the Dirac distribution at $x$, and
\[\mu_X = \Phi \mapsto \left( x \mapsto \sum_{\phi \in \mathrm{supp}(\Phi)} \Phi(\phi) \cdot \phi(x) \right).\]
It is presented by the theory of convex algebras which contains a binary operation $+_p$ for every $p\in (0,1)$ satisfying idempotence ($x+_p x= x$), skew commutativity ($x+_p y = y+_{1-p} x$) and skew associativity ($(x+_q y) +_p z = x+_{pq} (y +_{\frac{p (1-q)}{1-pq}} z)$). Here is the algebraic presentation of $\mD^{\mathrm{s}}$.
\begin{theorem}\label{thm:pressemibary}
    Let $\Sigma^{\mathrm{s}}_{\mD} = \{\a: 1, +_p :2 \mid p\in (0,1)\}$ and $E^{\mathrm{s}}_{\mD}$ contain 
    \begin{align*}
        \a \a x &= \a x\\
        \a(x+_p y) &= x+_p y\\
        \a x+_p \a y&= x+_p y\\
        x+_p x&= \a x\\
        x+_p y &= y+_{1-p} x\\
        (x+_q y) +_p z &= x+_{pq} (y +_{\frac{p (1-q)}{1-pq}} z),
    \end{align*}then $\Alg(\Sigma^{\mathrm{s}}_{\mD},E^{\mathrm{s}}_{\mD}) \cong \EMs(\mD)$
\end{theorem}
\begin{proof}
    In this proof, we will have to distinguish probability distributions seen in $\mD(X)$ and those seen in a $(\Sigma^{\mathrm{s}}_{\mD},E^{\mathrm{s}}_{\mD})$-algebra. In the former, we will write $px+\overline{p}y$ ($\overline{p}:=1-p$) in the binary case and $\sum_{i=1}^n p_i x_i$ in general. In the latter, we will write $x+_p y$ in the binary case and $\bigplus_{i=1}^n p_ix_i$ in general. Note that $\bigplus$ is well-defined by skew associativity of $+_p$ (assuming all $x_i$s are distinct). 

    The proof follows the sketch of the last two (all details are in the appendix). A $\mD$-semialgebra $a:\mD X \rightarrow X$ is sent to $(X,\a: 0, +_p: 2)$ where $\a := a(1-)$ and $+_p:= a(p-+\overline{p}-)$. A $(\Sigma^{\mathrm{s}}_{\mD},E^{\mathrm{s}}_{\mD})$-algebra $(X,\a:0, +_p: 2)$ is sent to $a: \mD X \rightarrow X = \sum_{i=1}^n p_i x_i \mapsto \bigplus_{i=1}^n p_i\a x_i$. Analogously to the proof for the semigroup monad, a step in the proof uses the fact that the equations in $E^{\mathrm{s}}_{\mD}$ imply
    \[\forall n\geq 2, \qquad \a\left( \bigplus_{i=1}^n p_ix_i \right) = \bigplus_{i=1}^n p_ix_i = \bigplus_{i=1}^n p_i \a x_i.\]
\end{proof}
These three examples show a clear link between the presentation of a monad on $\mathbf{Set}$ and of its semifree monad. They all fit in the following conjecture.
\begin{conjecture}\label{conj:pressemialg}
    Let $(M,\eta,\mu)$ be a monad on $\mathbf{Set}$ with an algebraic presentation $(\Sigma_M,E_M)$, then then the semifree monad $(\Ms,\etas,\mus)$ is presented by the signature $\Sigma_M \cup \{\a :1\}$ with the equations
    \begin{gather*}
        \a \a x = \a x\\
        \forall \mathsf{op}:n \in \Sigma_M, \a(\mathsf{op}(x_1,\dots, x_n)) = \mathsf{op}(x_1,\dots,x_n) = \mathsf{op}(\a x_1, \dots, \a x_n)\\
        \forall t(x_1,\dots, x_n) = s(x_1,\dots, x_n) \in E_M, t(\a x_1, \dots, \a x_n) = s(\a x_1, \dots, \a x_n).
    \end{gather*}
\end{conjecture}

\section{Directions for Future Work}\label{conclusion}
In this paper we proved that semialgebras for a monad $M$ on a
category with coproducts are in fact algebras for the semifree monad
$M^s$ with underlying functor $\id+M$. We also showed that weak
distributive laws $MT\Rightarrow TM$ correspond to certain strong
distributive laws $\Ms T\Rightarrow T\Ms$. 

Starting with a weak distributive law $\lambda: MT \Rightarrow TM$, we
now have two ways to obtain a weak composite of $T$ and $M$. If
idempotents split in the base category, Garner's method \cite{Garner}
yields a weak lifting $\widehat{T}$ of $T$ to $\EM(M)$ and hence a
monad structure on the functor $U^M \circ \widehat{T} \circ F^M$,
where $F^M$ is the free $M$-algebra functor sending $X$ to
$(MX,\mu^M_X)$. If all coproducts exist in the base category, our
method also yields a lifting $\widetilde{T}^{\mathrm{s}}$ of $T$ on
$\EM(\Ms)$ and hence a monad structure on the functor $U^{\Ms} \circ
\widetilde{T}^{\mathrm{s}} \circ F^{\Ms}$.

In general these composite monads are not the same, for example, for
the weak distributive law of~\cite{GoyPetrisan20} we obtain on one
hand the monad $\mathcal{P}_c\mathcal{D}$ of convex powersets of
distributions, and on the other, a monad
$\mathcal{P}(\id+\mathcal{D})$. Understanding how these composite
monads relate directly and how the second one can be used for the
semantics of computational effects is left for future investigations.


The adjunctions between $\EMs(M)$ and $\EM(M)$ described
in~\cite[Lemma 12]{Garner} can now be seen as adjunctions between
$\EM(M)$ and $\EM(\Ms)$ as drawn below.
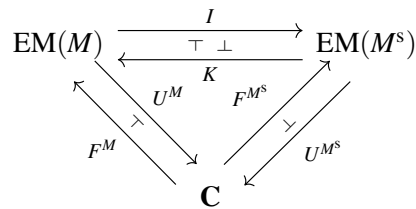
\begin{equation*}
        \begin{tikzcd}
            {\EM(M)} && {\EM(\Ms)} \\
            \\
            & {\mathbf{C}}
            \arrow[""{name=0, anchor=center, inner sep=0}, "{U^M}", shift left=2, from=1-1, to=3-2]
            \arrow[""{name=1, anchor=center, inner sep=0}, "{F^{\Ms}}", shift left=2, from=3-2, to=1-3]
            \arrow[""{name=2, anchor=center, inner sep=0}, "I", shift left=2, from=1-1, to=1-3]
            \arrow[""{name=3, anchor=center, inner sep=0}, "K", shift left=2, from=1-3, to=1-1]
            \arrow[""{name=4, anchor=center, inner sep=0}, "{U^{\Ms}}", shift left=2, from=1-3, to=3-2]
            \arrow[""{name=5, anchor=center, inner sep=0}, "{F^M}", shift left=2, from=3-2, to=1-1]
            \arrow["\dashv"{anchor=center, rotate=-45}, draw=none, from=1, to=4]
            \arrow["\dashv"{anchor=center, rotate=45}, draw=none, from=5, to=0]
            \arrow["\dashv"{anchor=center, rotate=90}, shift left=2, draw=none, from=3, to=2]
            \arrow["\dashv"{anchor=center, rotate=-90}, shift left=2, draw=none, from=2, to=3]
        \end{tikzcd}
\end{equation*}
One can show that $I$ is the functor induced by the monad map
$[\eta^M,\id_M]: \Ms \Rightarrow M$, and hence that $I \circ U^{M} =
U^{\Ms}$. Also, while this only implies $K \circ F^{\Ms}$ and $F^M$
are isomorphic, since both are left adjoints to $U^M$, one can also
prove they are equal. However, the other triangles do not necessarily
commute.


Another direction for future research is Conjecture
\ref{conj:pressemialg}. If it is resolved positively or if another
general algebraic presentation for semialgebras is discovered, we may
be able to use Theorem \ref{thm:wdlanddl} in conjunction with the very
general results of \cite{Zwart19} to obtain no-go theorems for
\textit{weak} distributive laws.
\bibliographystyle{eptcs}
\bibliography{refs}
\newpage
\section{Appendix}\label{appendix}
\subsection{Proofs for Section \ref{dllifts}}
\begin{proof}[Proof of Theorem \ref{thm:wdlliftings}]
    Let $\lambda :MT \Rightarrow TM$ be a weak distributive law, we define a functor $\widetilde{T}: \EMs(M) \rightarrow \EMs(M)$ that sends an $M$-semialgebra $MX \xrightarrow{x} X$ to $MTX \xrightarrow{\lambda_X} TMX \xrightarrow{Tx} TX$. First, we check that $Tx \circ \lambda_X$ is an $M$-semialgebra with the following diagram where (a) is \eqref{diag:dlmultM} instantiated for $\lambda$, (b) is the naturality of $\lambda$ and (c) is $T$ applied to associativity of $x$.
    \begin{equation}\label{diag:liftTmultalg}
    \begin{tikzcd}
    MMTX && MTX \\
    MTMX & TMMX & TMX \\
    MTX & TMX & TX
    \arrow["{\lambda_X}", from=1-3, to=2-3]
    \arrow["Tx", from=2-3, to=3-3]
    \arrow[""{name=0, anchor=center, inner sep=0}, "{\lambda_X}"', from=3-1, to=3-2]
    \arrow[""{name=1, anchor=center, inner sep=0}, "Tx"', from=3-2, to=3-3]
    \arrow[""{name=2, anchor=center, inner sep=0}, "{\mu^M_{TX}}", from=1-1, to=1-3]
    \arrow["{M\lambda_X}"', from=1-1, to=2-1]
    \arrow["MTx"', from=2-1, to=3-1]
    \arrow["TMx"', from=2-2, to=3-2]
    \arrow[""{name=3, anchor=center, inner sep=0}, "{T\mu^M_X}", from=2-2, to=2-3]
    \arrow[""{name=4, anchor=center, inner sep=0}, "{\lambda_{MX}}", from=2-1, to=2-2]
    \arrow["{\text{(a)}}"{description}, Rightarrow, draw=none, from=2, to=2-2]
    \arrow["{\text{(b)}}"{description}, shift right=1, Rightarrow, draw=none, from=4, to=0]
    \arrow["{\text{(c)}}"{description}, Rightarrow, draw=none, from=3, to=1]
    \end{tikzcd}
    \end{equation}
    Next, $\widetilde{T}$ sends a homomorphism $f: (X,x) \rightarrow (Y,y)$ to $Tf:(TX,Tx \circ \lambda_X) \rightarrow (TY,Ty\circ \lambda_Y)$ which is a homomorphism because
    \[Tf \circ Tx \circ \lambda_X = Ty \circ TMf \circ \lambda_X = Ty \circ \lambda_Y \circ MTf.\]
    We immediately see that $\forgets\widetilde{T} = T\forgets$. Next, we check that the components of the unit and multiplication determined by $\forgets \widetilde{\eta} = \eta^T\forgets$ and $\forgets\widetilde{\mu} = \mu^T\forgets$ are homomorphisms. The unit is $\widetilde{\eta}_{(X,x)} = \eta^T_X: (X,x) \rightarrow (TX,Tx \circ \lambda_X)$ which is a homomorphism because $Tx \circ \lambda_X \circ M\eta^T_X = Tx \circ \eta^T_{MX} = \eta^T_X \circ x$. The multiplication is $\widetilde{\mu}_{(X,x)} = \mu^T_X: (TTX, T(Tx \circ \lambda_X) \circ \lambda_{TX}) \rightarrow (TX, Tx \circ \lambda_X)$ which is a homomorphism by the following diagram where (a) is \eqref{diag:dlmultMhat} instantiated for $\lambda$ and (b) is naturality of $\mu^T$.
    \begin{equation}
        \begin{tikzcd}
            MTTX && MTX \\
            TMTX \\
            TTMX && TMX \\
            TTX && TX
            \arrow["{\lambda_{TX}}"', from=1-1, to=2-1]
            \arrow["{T\lambda_X}"', from=2-1, to=3-1]
            \arrow["TTx"', from=3-1, to=4-1]
            \arrow[""{name=0, anchor=center, inner sep=0}, "{\mu^T_X}"', from=4-1, to=4-3]
            \arrow["{\lambda_X}", from=1-3, to=3-3]
            \arrow["Tx", from=3-3, to=4-3]
            \arrow[""{name=1, anchor=center, inner sep=0}, "{M\mu^T_X}", from=1-1, to=1-3]
            \arrow[""{name=2, anchor=center, inner sep=0}, "{\mu^T_{MX}}", from=3-1, to=3-3]
            \arrow["{\text{(b)}}"{description}, Rightarrow, draw=none, from=2, to=0]
            \arrow["{\text{(a)}}"{description}, Rightarrow, draw=none, from=1, to=2]
        \end{tikzcd}
    \end{equation}
    We conclude that $(\widetilde{T},\widetilde{\eta},\widetilde{\mu})$ is a lifting of $T$ to $\EMs(M)$. It remains to show that $\widetilde{T}$ satisfies \eqref{eqn:hypothesislifting}. It follows from the following diagram where (a) is $T$ applied to \eqref{diag:unitmonad}, (b) is naturality of $\lambda$ and (c) is $T$ applied to associativity of $x$.
    \begin{equation}
    \begin{tikzcd}
        MTX && TMX \\
        & TMX \\
        MTMX & TMMX \\
        MTX & TMX & TX
        \arrow["{MT\eta^M_X}"', from=1-1, to=3-1]
        \arrow["MTx"', from=3-1, to=4-1]
        \arrow[""{name=0, anchor=center, inner sep=0}, "{\lambda_X}"', from=4-1, to=4-2]
        \arrow["Tx"', from=4-2, to=4-3]
        \arrow["{\lambda_X}", from=1-1, to=1-3]
        \arrow[""{name=1, anchor=center, inner sep=0}, "Tx", from=1-3, to=4-3]
        \arrow[""{name=2, anchor=center, inner sep=0}, "TMx", from=3-2, to=4-2]
        \arrow[""{name=3, anchor=center, inner sep=0}, "{T\mu^M_X}"', from=3-2, to=1-3]
        \arrow["{TM\eta^M_X}"', from=2-2, to=3-2]
        \arrow[""{name=4, anchor=center, inner sep=0}, "{\lambda_X}", from=1-1, to=2-2]
        \arrow[Rightarrow, no head, from=2-2, to=1-3]
        \arrow["{\text{(a)}}"{description}, Rightarrow, draw=none, from=2-2, to=3]
        \arrow["{\text{(c)}}"{description, pos=0.7}, Rightarrow, draw=none, from=2, to=1]
        \arrow["{\text{(b)}}"{description}, Rightarrow, draw=none, from=4, to=0]
    \end{tikzcd}
    \end{equation}

    In the other direction, we start with a lifting $(\widetilde{T},\widetilde{\eta},\widetilde{\mu})$ satisfying $\widetilde{T}x = \widetilde{T}x \circ MTx \circ MT\eta^M_X$ for any $M$-semialgebra $x:MX \rightarrow X$, and we let $\lambda_X = \widetilde{T}\mu^M_X \circ MT\eta^M_X$. We will show that $\lambda$ is a weak distributive law.

    First, naturality follows from the following diagram where (a) commutes by naturality of $\eta^M$ and (b) commutes because $TMf$ is the image of $Mf: MX \rightarrow MY$ which is a homomorphism between $\mu^M_X$ and $\mu^M_Y$ by naturality of $\mu^M$.
    \begin{equation}
        \begin{tikzcd}
            MTX & MTMX & TMX \\
            MTY & MTMY & TMY
            \arrow[""{name=0, anchor=center, inner sep=0}, "{MT\eta^M_Y}"', from=2-1, to=2-2]
            \arrow[""{name=1, anchor=center, inner sep=0}, "{\widetilde{T}\mu^M_Y}"', from=2-2, to=2-3]
            \arrow["MTf"', from=1-1, to=2-1]
            \arrow[""{name=2, anchor=center, inner sep=0}, "{MT\eta^M_X}", from=1-1, to=1-2]
            \arrow[""{name=3, anchor=center, inner sep=0}, "{\widetilde{T}\mu^M_X}", from=1-2, to=1-3]
            \arrow["MTMf"{description}, from=1-2, to=2-2]
            \arrow["TMf", from=1-3, to=2-3]
            \arrow["{\text{(a)}}"{description}, Rightarrow, draw=none, from=2, to=0]
            \arrow["{\text{(b)}}"{description}, Rightarrow, draw=none, from=3, to=1]
        \end{tikzcd}
    \end{equation}
    Next, we show that the instances of \eqref{diag:dlunitMhat}, \eqref{diag:dlmultM} and \eqref{diag:dlmultMhat} for $\lambda$ commute by paving the following diagrams.
    \begin{gather}
        \begin{tikzcd}[ampersand replacement=\&]
        \&\& MX \\
        \\
        \&\& MMX \& MX \\
        MTX \&\& MTMX \&\& TMX
        \arrow[""{name=0, anchor=center, inner sep=0}, "{M\eta^T_X}"', curve={height=30pt}, from=1-3, to=4-1]
        \arrow["{MT\eta^M_X}"', from=4-1, to=4-3]
        \arrow[""{name=1, anchor=center, inner sep=0}, "{\widetilde{T}\mu^M_X}"', from=4-3, to=4-5]
        \arrow["{M\eta^M_X}"', from=1-3, to=3-3]
        \arrow["{M\eta^T_{MX}}"', from=3-3, to=4-3]
        \arrow["{\eta^T_{MX}}", curve={height=-30pt}, from=1-3, to=4-5]
        \arrow[""{name=2, anchor=center, inner sep=0}, "{\mu^M_X}"', from=3-3, to=3-4]
        \arrow["{\eta^T_{MX}}", from=3-4, to=4-5]
        \arrow[""{name=3, anchor=center, inner sep=0}, Rightarrow, no head, from=1-3, to=3-4]
        \arrow["{\text{(a)}}"{description}, Rightarrow, draw=none, from=0, to=4-3]
        \arrow["{\text{(c)}}"{description, pos=0.6}, Rightarrow, draw=none, from=2, to=1]
        \arrow["{\text{(b)}}"{description}, Rightarrow, draw=none, from=3, to=3-3]
        \end{tikzcd}\\
        \begin{tikzcd}[ampersand replacement=\&]
            MMTX \& MMTMX \& MTMX \& MTMMX \& TMMX \\
            \& MMTMX \&\& MTMX \\
            MTX \&\& MTMX \&\& TMX
            \arrow["{MT\eta^M_X}"', from=3-1, to=3-3]
            \arrow[""{name=0, anchor=center, inner sep=0}, "{\widetilde{T}\mu^M_X}"', from=3-3, to=3-5]
            \arrow["{MMT\eta^M_X}", from=1-1, to=1-2]
            \arrow["{M\widetilde{T}\mu^M_X}", from=1-2, to=1-3]
            \arrow["{MT\eta^M_{MX}}", from=1-3, to=1-4]
            \arrow[""{name=1, anchor=center, inner sep=0}, "{\widetilde{T}\mu^M_{MX}}", from=1-4, to=1-5]
            \arrow["{\mu^M_{TX}}"', from=1-1, to=3-1]
            \arrow["{T\mu^M_X}", from=1-5, to=3-5]
            \arrow["{MMT\eta^M_X}"', from=1-1, to=2-2]
            \arrow["{\mu^M_{TMX}}"', from=2-2, to=3-3]
            \arrow[""{name=2, anchor=center, inner sep=0}, "{\widetilde{T}\mu^M_X}"', from=2-4, to=3-5]
            \arrow[""{name=3, anchor=center, inner sep=0}, "{M\widetilde{T}\mu^M_X}", from=2-2, to=2-4]
            \arrow["{MT\mu^M_X}", from=1-4, to=2-4]
            \arrow[""{name=4, anchor=center, inner sep=0}, Rightarrow, no head, from=1-3, to=2-4]
            \arrow["{\text{(f)}}"{description}, Rightarrow, draw=none, from=2-2, to=3-1]
            \arrow["{\text{(g)}}"{description}, Rightarrow, draw=none, from=3, to=0]
            \arrow["{\text{(d)}}"{description, pos=0.7}, shift right=1, Rightarrow, draw=none, from=4, to=1-4]
            \arrow["{\text{(e)}}"{description}, Rightarrow, draw=none, from=1, to=2]
        \end{tikzcd}\\
        \begin{tikzcd}[ampersand replacement=\&]
            MTTX \& MTMTX \& TMTX \& TMTMX \& TTMX \\
            \&\& MTMTMX \& MTTMX \\
            \\
            \& MTTMX \&\& TTMX \\
            MTX \& MTMX \&\&\& TMX
            \arrow[""{name=0, anchor=center, inner sep=0}, "{M\mu^T_X}"', from=1-1, to=5-1]
            \arrow["{MT\eta^M_{TX}}", from=1-1, to=1-2]
            \arrow["{\widetilde{T}\mu^M_{TX}}", from=1-2, to=1-3]
            \arrow["{TMT\eta^M_X}", from=1-3, to=1-4]
            \arrow["{T\widetilde{T}\mu^M_X}", from=1-4, to=1-5]
            \arrow["{\mu^T_{MX}}", from=1-5, to=5-5]
            \arrow["{MT\eta^M_X}"', from=5-1, to=5-2]
            \arrow[""{name=1, anchor=center, inner sep=0}, "{\widetilde{T}\mu^M_X}"', from=5-2, to=5-5]
            \arrow["{MTT\eta^M_X}", from=1-1, to=4-2]
            \arrow["{M\mu^T_{MX}}"', from=4-2, to=5-2]
            \arrow["{MTMT\eta^M_X}"', from=1-2, to=2-3]
            \arrow["{\widetilde{T}\mu^M_{TMX}}"{description}, from=2-3, to=1-4]
            \arrow["{MT\eta^M_{TMX}}", from=4-2, to=2-3]
            \arrow[""{name=2, anchor=center, inner sep=0}, "{\widetilde{T}\widetilde{T}\mu^M_X}", from=4-2, to=4-4]
            \arrow["{\mu^T_{MX}}", from=4-4, to=5-5]
            \arrow["{\text{(h)}}"{description, pos=0.4}, shift left=2, Rightarrow, draw=none, from=1-2, to=4-2]
            \arrow["{\text{(i)}}"{description}, Rightarrow, draw=none, from=1-3, to=2-3]
            \arrow[Rightarrow, no head, from=4-4, to=1-5]
            \arrow["{MT\widetilde{T}\mu^M_X}"', from=2-3, to=2-4]
            \arrow["{\widetilde{T}\widetilde{T}\mu^M_X}"{description}, from=2-4, to=1-5]
            \arrow["{\text{(j)}}"{description}, Rightarrow, draw=none, from=1-4, to=2-4]
            \arrow["{\text{(m)}}"{description}, Rightarrow, draw=none, from=2, to=1]
            \arrow["{\text{(k)}}"{description}, Rightarrow, draw=none, from=0, to=4-2]
            \arrow["{\text{(l)}}"{description}, Rightarrow, draw=none, from=2-4, to=2]
        \end{tikzcd}
    \end{gather}
    \begin{multicols}{2}
        \begin{enumerate}[(a)]
            \item Naturality of $\eta^M$ and $\eta^T$.
            \item By left of \eqref{diag:unitmonad}.
            \item By hypothesis, $\widetilde{\eta}_{(MX,\mu^M_X)} = \eta^T_{MX}$ is a homomorphism.
            \item Apply $MT$ to right of \eqref{diag:unitmonad}.
            \item Apply $\widetilde{T}$ to $\mu^M_X$ as a homomorphism $(MMX,\mu^M_{MX}) \rightarrow (MX,\mu^M_X)$ in $\EMs(M)$.
            \item Naturality of $\mu^M$ and $\eta^M$.
            \item Associativity of $\widetilde{T}\mu^M_X$.
            \item Naturality of $\eta^M$.
            \item Apply $\widetilde{T}$ to $MT\eta^M_X$ as a homomorphism $(MTX, \mu^M_{TX}) \rightarrow (MTMX, \mu^M_{TMX})$ in $\EMs(M)$.
            \item Apply $\widetilde{T}$ to $\widetilde{T}\mu^M_X$ as a homomorphism $(MTMX, \mu^M_{TMX}) \rightarrow (TMX,\widetilde{T}\mu^M_X)$ in $\EMs(M)$.
            \item Naturality of $\mu^T$ and $\eta^M$.
            \item By \eqref{eqn:hypothesislifting}.
            \item By hypothesis, $\widetilde{\mu}_{(MX,\mu^M_X)} = \mu^T_{MX}$ is a homomorphism.
        \end{enumerate}
    \end{multicols}
    We conclude that $\lambda$ is a weak distributive law.

    Finally, we are left to show that the operations we described are inverses. Starting with $\lambda$, we obtain a lifting $\widetilde{T}$ sending $\mu^M_X$ to $T\mu^M_X \circ \lambda_{MX}$ which is sent to $\lambda'$ whose component at $X$ is 
    \[\lambda'_X = T\mu^M_X \circ \lambda_{MX} \circ MT\eta^M_X = T\mu^M_X \circ TM\eta^M_X \circ \lambda_X = \lambda_X.\]
    Starting with $\widetilde{T}$, we obtain a weak distributive law $\lambda$ whose component at $X$ is $\widetilde{T}\mu^M_X \circ MT\eta^M_X$ which is sent to $\widetilde{T}'$ which sends an $M$-semialgebra $x:MX \rightarrow X$ to 
    \begin{align*}
        \widetilde{T}'x &= Tx \circ \lambda_X\\
        &= Tx \circ \widetilde{T}\mu^M_X \circ MT\eta^M_X\\
        &= \widetilde{T}x \circ MTx \circ MT\eta^M_X &&\text{Apply $\widetilde{T}$ to $x$ as a homomorphism $(MX,\mu^M_X) \rightarrow (X,x)$.}\\
        &= \widetilde{T}x. &&\text{$\widetilde{T}$ satisfies \eqref{eqn:hypothesislifting}}
    \end{align*}
    The theorem follows.
\end{proof}

\subsection{Proofs for Section \ref{exmps}}
\begin{proof}[Proof of Theorem \ref{thm:pressemisemigroup}]
    Given a $(-)^+$-semialgebra $a:X^+ \rightarrow X$, we define $\a:= a[-] :X \rightarrow X$ and $-\cdot- := a[-,-] : X\times X \rightarrow X$. Let us verify each of the equations in $E^{\mathrm{s}}_+$ hold.
    \begin{multicols}{3}
    \begin{align*}
        \a\a x &= a[\a x]\\
        &= a[a[x]]\\
        &= (a \circ a^+)[[x]]\\
        &= (a \circ \mu_X)[[x]]\\
        &= a[x]\\
        &= \a x
    \end{align*}
    \begin{align*}
        \a(x \cdot y) &= a[a[x,y]]\\
        &= (a \circ a^+)[[x,y]]\\
        &= (a \circ \mu_X)[[x,y]]\\
        &= a[x,y]= x \cdot y
    \end{align*}
    \begin{align*}
        \a x \cdot \a y &= a[a[x],a[y]]\\
        &= (a \circ a^+)[[x],[y]]\\
        &= (a \circ \mu_X)[[x],[y]]\\
        &= a[x,y] = x \cdot y
    \end{align*}
    \begin{align*}
        (x \cdot y) \cdot z &= a[a[x,y],z]\\
        &= a[a[a[x,y]],a[z]]\\
        &= a[a[x,y],a[z]]\\
        &= (a \circ a^+)[[x,y],[z]]\\
        &= (a \circ \mu_X)[[x,y],[z]]\\
        &= a[x,y,z]\\
        &= \vdots \text{ symmetric argument}\\
        &= a[x,a[y,z]] = x \cdot (y \cdot z).
    \end{align*}
\end{multicols}

    Conversely, given $\a: X \rightarrow X$ and $\cdot: X\times X \rightarrow X$ satisfying the equations in $E^{\mathrm{s}}_+$, we define $a:X^+ \rightarrow X$ by $[x_1,\dots,x_n] \mapsto \a x_1 \cdots \a x_n$ which is well-defined by associativity. We need to generalize the equation $\a x \cdot \a y = x \cdot y = \a(x \cdot y)$ to longer strings of $\cdot$, namely, we claim that $\a x_1 \cdots \a x_n = x_1 \cdots x_n = \a(x_1\cdots x_n)$. We proceed by induction starting with $n= 2$ which is true by hypothesis. If it holds for $n-1$, then 
    \begin{align*}
        \a x_1 \cdots \a x_n &= (\a x_1 \cdots \a x_{n-1}) \cdot \a x_n\\
        &= (x_1 \cdots x_{n-1}) \cdot \a x_n\\
        &= \a (x_1 \cdots x_{n-1}) \cdot \a \a x_n\\
        &= \a(x_1 \cdots x_{n-1}) \cdot \a x_n = \a(x_1\cdots x_n)\\
        &= (x_1 \cdots x_{n-1}) \cdot x_n\\
        &= x_1 \cdots x_n.
    \end{align*}
    In the following derivation which shows $a$ is a $(-)^+$-semialgebra, we need a slightly weaker version that holds even for $n=1$: $\a(\a x_1 \cdots \a x_n) = \a x_1 \cdots \a x_n$. For any $L = [[x_{1,1},\dots, x_{n_1,1}], \dots, [x_{k,1},\dots, x_{n_k,k}]] \in (X^+)^+$,
    \begin{align*}
        a(a^+(L)) &= a[\a x_{1,1}\cdots \a x_{n_1,1}, \dots, \a x_{k,1}\cdots \a x_{n_k,k}]\\
        &= \a(\a x_{1,1}\cdots \a x_{n_1,1})\cdots \a(\a x_{k,1}\cdots \a x_{n_k,k})\\
        &= \a x_{1,1}\cdots \a x_{n_1,1} \cdots \a x_{k,1}\cdots \a x_{n_k,k}\\
        &= a[x_{1,1},\dots, x_{n_1,1}, x_{2,1} \dots, x_{n_{k-1},k-1},x_{k,1},\dots, x_{n_k,k}]\\
        &= a(\mu_X(L)).
    \end{align*}
    
    Let us show these operations are inverses. If $\a$ and $\cdot$ are obtained from $a:X^+ \rightarrow X$, then we proceed by induction to show that $a[x_1,\dots, x_n] = \a x_1 \cdots \a x_n$. For $n = 1$, it is clear. For $n=2$, we have $a[x_1,x_2] = a[a[x_1],a[x_2]] = \a x_1 \cdot \a x_2$. Suppose it holds for $n-1$, then 
    \begin{align*}
        a[x_1,\dots, x_n] &= (a \circ \mu_X)[[x_1,\dots, x_{n-1}], [x_n]]\\
        &= a[a[x_1,\dots,x_{n-1}],a[x_n]]\\
        &= a[\a x_1\cdots \a x_{n-1},\a x_n]\\
        &= \a(\a x_1\cdots \a x_{n-1}) \cdot \a \a x_n\\
        &= \a x_1 \cdots \a x_n.
    \end{align*}
    We conclude that the semialgebra obtained from $\a$ and $\cdot$ is $a$. In the other direction, let $a$ be obtained from $\a$ and $\cdot$. We have $a[x,y] = \a x \cdot \a y = x \cdot y$ and $a[x] = \a x$ showing that the operations obtained from $a$ are $\a$ and $\cdot$.

    Finally, we check that these operations are functorial. Suppose $f:X \rightarrow Y$ is a homomorphism between $a$ and $b$, then 
    \begin{gather*}
        f(\a(x)) = f(a[x]) = b[f(x)] = \b f(x)\\
        f(x\cdot^a y) = f(a[x,y]) = b[f(x),f(y)] = f(x) \cdot^b f(y).
    \end{gather*}
    Conversely, if $f \circ \a = \b \circ f$ and $f \circ \cdot^a = \cdot^b \circ (f\times f)$, then 
    \begin{gather*}
        f(a[x_1,\dots,x_n]) = f(\a x_1 \cdot^a\cdot^a\cdot^a \a x_n) = \b f(x_1) \cdot^b\cdot^b\cdot^b \b f(x_n) = b[f(x_1),\dots,f(x_n)].
    \end{gather*}
\end{proof}
\begin{proof}[Proof of Theorem \ref{thm:pressemibary}]
    Given a $\mD$-semialgebra $a:\mD X \rightarrow X$, we define $x+_p y := a(px+\overline{p}y)$ and $\a x = a(1x)$. Let us verify each equation in $E^{\mathrm{s}}_{\mD}$ holds.
    \begin{multicols}{3}
        \begin{align*}
            \a \a x &= \a a(1x)\\
            &= a(1a(1x))\\
            &= (a \circ \mD a)(1(1x))\\
            &= (a \circ \mu_X)(1(1x))\\
            &= a(1x) = \a x
        \end{align*}
        \begin{align*}
            x+_p x &= a(px+\overline{p}x)\\
            &= a(1x) = \a x
        \end{align*}
        \begin{align*}
            x+_p y &= a(px + \overline{p}y)\\
            &= a(\overline{p}y + px)\\
            &= y+_{1-p}x
        \end{align*}
        \begin{align*}
            \a (x+_p y) &= a(1a(px+\overline{p}y))\\
            &= (a \circ \mD a)(1(px+\overline{p}y))\\
            &= (a \circ \mu_X)(1(px+\overline{p}y))\\
            &= a(px+\overline{p}y)\\
            &= x+_p y
        \end{align*}
        \begin{align*}
            \a x+_p \a y &= a(p\a x+\overline{p}\a y)\\
            &= a(pa(1x)+\overline{p}a(1y))\\
            &= (a \circ \mD a)(p(1x)+\overline{p}(1y))\\
            &= (a \circ \mu_X)(p(1x)+\overline{p}(1y)))\\
            &= a(px+\overline{p}y)\\
            &= x+_p y
        \end{align*}
        \begin{align*}
            (x+_q y) +_p z &= a(p(x+_q y) + \overline{p}z)\\
            &= a(pa(qx + \overline{q}y)) + \overline{p}z)\\
            &= a(pa(1a(qx+\overline{q}y))+\overline{p}a(1z))\\
            &= a(pa(qx+\overline{q}y)+\overline{p}a(1z))\\
            &= (a \circ \mD a)(p(qx+\overline{q}y)+\overline{p}(1z))\\
            &= (a \circ \mu_X)(p(qx+\overline{q}y)+\overline{p}(1z))\\
            &= a(pqx+ p\overline{q}y + \overline{p}z)\\
            &= \vdots \text{ symmetric argument}\\
            &= x+_{pq} (y +_{\frac{p (1-q)}{1-pq}} z)
        \end{align*}
    \end{multicols}

    Conversely, given $+_p$ and $\a$ satisfying $E^{\mathrm{s}}_{\mD}$, we define $a:\mD X \rightarrow X$ by $a(\sum_{i=1}^n p_i x_i) = \bigplus_{i=1}^n p_i \a x_i$. We need to generalize the equation $\a(x+_p y) = x+_p y = \a x+_p \a y$ to distributions with larger support, namely, we claim that $\a\left( \bigplus_{i=1}^n p_ix_i \right) = \bigplus_{i=1}^n p_ix_i = \bigplus_{i=1}^n p_i \a x_i$. We proceed by induction starting with $n=2$ which is true by hypothesis. Suppose it holds for $n-1$, and let $p'_i$ for $1\leq i \leq n$ be such that $\bigplus_{i=1}^n p_i x_i = \bigplus_{i=1}^{n-1} p'_i x_i +_{p'_n} x_n$, then 
    \begin{align*}
        \bigplus_{i=1}^n p_i \a x_i &= \bigplus_{i=1}^{n-1} p'_i \a x_i +_{p'_n} \a x_n\\
        &= \bigplus_{i=1}^{n-1} p'_i x_i +_{p'_n} \a x_n\\
        &= \a\left( \bigplus_{i=1}^{n-1} p'_i x_i \right) +_{p'_n} \a\a x_n\\
        &= \a\left( \bigplus_{i=1}^{n-1} p'_i x_i \right) +_{p'_n} \a x_n = \a\left( \bigplus_{i=1}^{n-1} p'_i x_i +_{p'_n} x_n \right) = \a\left( \bigplus_{i=1}^np_i x_i \right)\\
        &= \left( \bigplus_{i=1}^{n-1} p'_i x_i \right) +_{p'_n} x_n\\
        &= \bigplus_{i=1}^np_i x_i.
    \end{align*}
    In the following derivation which shows that $a$ is a $\mD$-semialgebra, we need a slightly weaker version that holds even for $n=1$: $\a(\bigplus_{i=1}^n p_i \a x_i) = \bigplus_{i=1}^n p_i \a x_i$. For any $\Phi = \sum_{i=1}^n p_i \left( \sum_{j=1}^{m_i} q_{i,j}x_{i,j}\right)$,\footnote{In this derivation, we cannot assume that all the $x_{i,j}$s are distinct, so skew associativity is not enough to say that $\bigplus_{i=1}^n$ is well-defined. However, we may use skew commutativity and idempotence as well because $\a$ is applied to every term of the sum and idempotence holds when this is the case.}
    \begin{align*}
        a(\mD a(\Phi)) &= a\left( \sum_{i=1}^np_i\left( \bigplus_{j=1}^{m_i} q_{i,j}\a x_{i,j} \right) \right)\\
        &= \bigplus_{i=1}^n p_i \a \left( \bigplus_{j=1}^{m_i} q_{i,j}\a x_{i,j} \right)\\
        &= \bigplus_{i=1}^n p_i \left( \bigplus_{j=1}^{m_i} q_{i,j}\a x_{i,j} \right)\\
        &= \bigplus_{i=1}^n \bigplus_{j=1}^{m_i}p_iq_{i,j}\a x_{i,j}\\
        &= a\left( \sum_{i=1}^n\sum_{j=1}^{m_i} p_iq_{i,j}x_{i,j} \right)\\
        &= a(\mu_X(\Phi)).
    \end{align*}

    Let us show that these operations are inverses. If $\a$ and $+_p$ are obtained from $a: \mD X \rightarrow X$, then we proceed by induction to show that $a(\sum_{i=1}^n p_ix_i) = \bigplus_{i=1}^n p_i \a x_i$. For $n=1$ it is clear. For $n=2$, we have $a(px+\overline{p}y) = a(pa(1x)+\overline{p}a(1y)) = \a x +_p \a y$. Suppose it hods for $n-1$, then
    \begin{align*}
        a\left( \sum_{i=1}^n p_i x_i \right) &= (a \circ \mu_X)\left( p'_n\left( \sum_{i=1}^{n-1}\frac{p'_i}{p'_n}x_i \right) +\overline{p'_n} (1x_n)\right)\\
        &= (a \circ \mD a)\left( p'_n\left( \sum_{i=1}^{n-1}\frac{p'_i}{p'_n}x_i \right) +\overline{p'_n} (1x_n)\right)\\
        &= a\left( p'_na\left( \sum_{i=1}^{n-1}\frac{p'_i}{p'_n}x_i \right) +\overline{p'_n} a(1x_n)\right)\\
        &= a\left(p'_n\left( \bigplus_{i=1}^{n-1}\frac{p'_i}{p'_n}\a x_i \right) + \overline{p'_n} \a x_n  \right)\\
        &= \a \left( \bigplus_{i=1}^{n-1}\frac{p'_i}{p'_n}\a x_i \right) +_{p'_n} \a \a x_n\\
        &= \left( \bigplus_{i=1}^{n-1}\frac{p'_i}{p'_n}\a x_i \right) +_{p'_n} \a x_n\\
        &= \bigplus_{i=1}^{n} \a x_i.
    \end{align*}
    We conclude that the semialgebra obtained from $\a$ and $+_p$ is $a$.

    In the other direction, let $a$ be obtained from $\a$ and $+_p$. We have $a(px+\overline{p}y) = \a x+_p \a y = x+_p y$ and $a(1x) = \a x$ showing that the operations obtained from $a$ are $\a$ and $+_p$. 

    Finally, we check that these operations are functorial. Suppose $f:X \rightarrow Y$ is a homomorphism between $a$ and $b$, then 
    \begin{gather*}
        f(\a x) = f(a(1x)) = b(1f(x)) = \b f(x)\\
        f(x+_p^a y) = f(a(px+\overline{p}y)) = b(pf(x)+\overline{p}f(y)) = f(x)+_p^b f(y).
    \end{gather*}
    Conversely, if $f \circ \a = \b \circ f$ and $f \circ +_p^a = +_p^b \circ (f\times f)$, then
    \[f(a(\sum_{i=1}^np_i x_i)) = f\left( \bigplus_{i=1}^n p_i \a x_i \right) =  \bigplus_{i=1}^n p_i f(\a x_i) =  \bigplus_{i=1}^n p_i \b f(x_i) = b(\sum_{i=1}^n p_i f(x_i)).\]
\end{proof}
\end{document}